\documentclass[a4paper,12pt]{amsart}

\usepackage{a4wide,graphicx,amssymb,amscd,amsmath,latexsym,amsbsy,amsthm,color,multicol}

\usepackage{marginnote,mathtools,epsfig,enumitem,array,calc,rotating,bm,bbm,mathrsfs} 
\usepackage{subfig}
\usepackage[libertine,cmintegrals,cmbraces,vvarbb]{newtxmath}
\usepackage{booktabs}
\usepackage{multirow}
\usepackage{breqn}
\usepackage{etoolbox}

\usepackage{tikz}

\usetikzlibrary{decorations.markings}
\usetikzlibrary{calc} 
\usetikzlibrary{arrows} 
\usetikzlibrary{patterns} 
\usetikzlibrary{decorations.pathreplacing} 

\usepackage[
   hmarginratio={1:1},     
   vmarginratio={1:1},     
   textwidth=457pt,        
   textheight=630pt,       
   heightrounded,          
]{geometry}

\makeatletter
\let\old@tocline\@tocline
\let\section@tocline\@tocline

\newcommand{\subsection@dotsep}{4.5}
\newcommand{\subsubsection@dotsep}{4.5}
\patchcmd{\@tocline}
  {\hfil}
  {\nobreak
     \leaders\hbox{$\m@th
        \mkern \subsection@dotsep mu\hbox{.}\mkern \subsection@dotsep mu$}\hfill
     \nobreak}{}{}
\let\subsection@tocline\@tocline
\let\@tocline\old@tocline

\patchcmd{\@tocline}
  {\hfil}
  {\nobreak
     \leaders\hbox{$\m@th
        \mkern \subsubsection@dotsep mu\hbox{.}\mkern \subsubsection@dotsep mu$}\hfill
     \nobreak}{}{}
\let\subsubsection@tocline\@tocline
\let\@tocline\old@tocline

\let\old@l@subsection\l@subsection
\let\old@l@subsubsection\l@subsubsection

\def\@tocwriteb#1#2#3{%
  \begingroup
    \@xp\def\csname #2@tocline\endcsname##1##2##3##4##5##6{%
      \ifnum##1>\c@tocdepth
      \else \sbox\z@{##5\let\indentlabel\@tochangmeasure##6}\fi}%
    \csname l@#2\endcsname{#1{\csname#2name\endcsname}{\@secnumber}{}}%
  \endgroup
  \addcontentsline{toc}{#2}%
    {\protect#1{\csname#2name\endcsname}{\@secnumber}{#3}}}%

\newlength{\@tocsectionindent}
\newlength{\@tocsubsectionindent}
\newlength{\@tocsubsubsectionindent}
\newlength{\@tocsectionnumwidth}
\newlength{\@tocsubsectionnumwidth}
\newlength{\@tocsubsubsectionnumwidth}
\newcommand{\settocsectionnumwidth}[1]{\setlength{\@tocsectionnumwidth}{#1}}
\newcommand{\settocsubsectionnumwidth}[1]{\setlength{\@tocsubsectionnumwidth}{#1}}
\newcommand{\settocsubsubsectionnumwidth}[1]{\setlength{\@tocsubsubsectionnumwidth}{#1}}
\newcommand{\settocsectionindent}[1]{\setlength{\@tocsectionindent}{#1}}
\newcommand{\settocsubsectionindent}[1]{\setlength{\@tocsubsectionindent}{#1}}
\newcommand{\settocsubsubsectionindent}[1]{\setlength{\@tocsubsubsectionindent}{#1}}

\renewcommand{\l@section}{\section@tocline{1}{\@tocsectionvskip}{\@tocsectionindent}{\@tocsectionnumwidth}{\@tocsectionformat}}%
\renewcommand{\l@subsection}{\subsection@tocline{1}{\@tocsubsectionvskip}{\@tocsubsectionindent}{\@tocsubsectionnumwidth}{\@tocsubsectionformat}}%
\renewcommand{\l@subsubsection}{\subsubsection@tocline{1}{\@tocsubsubsectionvskip}{\@tocsubsubsectionindent}{\@tocsubsubsectionnumwidth}{\@tocsubsubsectionformat}}%
\newcommand{\@tocsectionformat}{}
\newcommand{\@tocsubsectionformat}{}
\newcommand{\@tocsubsubsectionformat}{}
\expandafter\def\csname toc@1format\endcsname{\@tocsectionformat}
\expandafter\def\csname toc@2format\endcsname{\@tocsubsectionformat}
\expandafter\def\csname toc@3format\endcsname{\@tocsubsubsectionformat}
\newcommand{\settocsectionformat}[1]{\renewcommand{\@tocsectionformat}{#1}}
\newcommand{\settocsubsectionformat}[1]{\renewcommand{\@tocsubsectionformat}{#1}}
\newcommand{\settocsubsubsectionformat}[1]{\renewcommand{\@tocsubsubsectionformat}{#1}}
\newlength{\@tocsectionvskip}
\newcommand{\settocsectionvskip}[1]{\setlength{\@tocsectionvskip}{#1}}
\newlength{\@tocsubsectionvskip}
\newcommand{\settocsubsectionvskip}[1]{\setlength{\@tocsubsectionvskip}{#1}}
\newlength{\@tocsubsubsectionvskip}
\newcommand{\settocsubsubsectionvskip}[1]{\setlength{\@tocsubsubsectionvskip}{#1}}

\patchcmd{\tocsection}{\indentlabel}{\makebox[\@tocsectionnumwidth][l]}{}{}
\patchcmd{\tocsubsection}{\indentlabel}{\makebox[\@tocsubsectionnumwidth][l]}{}{}
\patchcmd{\tocsubsubsection}{\indentlabel}{\makebox[\@tocsubsubsectionnumwidth][l]}{}{}

\newcommand{\@sectypepnumformat}{}
\renewcommand{\contentsline}[1]{%
  \expandafter\let\expandafter\@sectypepnumformat\csname @toc#1pnumformat\endcsname%
  \csname l@#1\endcsname}
\newcommand{\@tocsectionpnumformat}{}
\newcommand{\@tocsubsectionpnumformat}{}
\newcommand{\@tocsubsubsectionpnumformat}{}
\newcommand{\setsectionpnumformat}[1]{\renewcommand{\@tocsectionpnumformat}{#1}}
\newcommand{\setsubsectionpnumformat}[1]{\renewcommand{\@tocsubsectionpnumformat}{#1}}
\newcommand{\setsubsubsectionpnumformat}[1]{\renewcommand{\@tocsubsubsectionpnumformat}{#1}}
\renewcommand{\@tocpagenum}[1]{%
  \hfill {\mdseries\@sectypepnumformat #1}}

\let\oldappendix\appendix
\renewcommand{\appendix}{%
  \leavevmode\oldappendix%
  \addtocontents{toc}{%
    \protect\settowidth{\protect\@tocsectionnumwidth}{\protect\@tocsectionformat\sectionname\space}%
    \protect\addtolength{\protect\@tocsectionnumwidth}{2em}}%
}
\makeatother

\makeatletter
\settocsectionnumwidth{2em}
\settocsubsectionnumwidth{2.5em}
\settocsubsubsectionnumwidth{3em}
\settocsectionindent{1pc}%
\settocsubsectionindent{\dimexpr\@tocsectionindent+\@tocsectionnumwidth}%
\settocsubsubsectionindent{\dimexpr\@tocsubsectionindent+\@tocsubsectionnumwidth}%
\makeatother

\settocsectionvskip{0pt}
\settocsubsectionvskip{-5pt}
\settocsubsubsectionvskip{-5pt}

\settocsectionformat{\bfseries}
\settocsubsectionformat{\mdseries}
\settocsubsubsectionformat{\mdseries}
\setsectionpnumformat{\bfseries}
\setsubsectionpnumformat{\mdseries}
\setsubsubsectionpnumformat{\mdseries}

\let\oldtableofcontents\tableofcontents
\renewcommand{\tableofcontents}{%
  \vspace*{-\linespacing}
  \oldtableofcontents}

\numberwithin{equation}{section}
\theoremstyle{plain}
 
\newtheorem{thm}{Theorem}[section]

\newtheorem{prop}[thm]{Proposition}

\newtheorem{defi}[thm]{Definition}
\newtheorem{lem}[thm]{Lemma}

\newtheorem{eg}[thm]{Example}

\theoremstyle{remark}
\newtheorem{rema}[thm]{Remark}
\newtheorem{rmk}[thm]{Remark}

\newcommand{\Z}{\mathbb{Z}}

\newcommand{\C}{\mathbb{C}}

\newcommand{\J}{\mathbb{J}}

\newcommand{\dd}{\mathrm{d}}
\newcommand{\DD}{\mathrm{D}}
\newcommand{\ii}{\mathrm{i}}

\newcommand{\hypref}[2]{\ifx\href\asklfhas #2\else\href{#1}{#2}\fi}
\newcommand{\Secref}[1]{Section~\ref{#1}}

\newcommand{\Appref}[1]{Appendix~\ref{#1}}

\renewcommand{\eqref}[1]{(\ref{#1})}

\def\[{\begin{equation}}
\def\]{\end{equation}}
\def\<{\begin{eqnarray}}
\def\>{\end{eqnarray}}

\title[]{Six-vertex model and non-linear \\ differential equations II. Continuous symmetries}
\author{W. Galleas}

\address{Institut f\"ur Theoretische Physik, Eidgen\"ossische Technische Hochschule Z\"urich, Wolfgang-Pauli-Strasse 27, 8093 Z\"urich, Switzerland}
\email{galleasw@phys.ethz.ch}

\thanks{The work of W.G. is partially supported by the Swiss National Science Foundation through the NCCR SwissMAP}

\begin{document}

\begin{abstract}

This paper is a continuation of our previous work \emph{Six-vertex model and non-linear differential equations I. Spectral problem} in which we have put forward a method for studying the spectrum of the six-vertex model based on non-linear differential equations. Here we intend to elaborate on that approach and also discuss properties of the spectrum unveiled by the aforementioned differential formulation of the transfer matrix's eigenvalue problem. In particular, we intend to demonstrate how this differential approach allows one to study
\emph{continuous symmetries} of the transfer matrix's spectrum through the \emph{Lie groups method}.

\end{abstract}

\maketitle

\tableofcontents

\section{Introduction} \label{sec:INTRO}

Non-linear differential equations and methods for solving them exactly constitute a fundamental part of the modern theory of integrable systems. Although the presence of such structures in \emph{Exactly Solvable Models of Statistical Mechanics} is not apparent at first sight, there exist a large number of results suggesting the same structures governing \emph{classical integrable systems} should also be present in the context of Statistical Mechanics. See for instance \cite{Barouch_1973, Tracy_1973, Wu_1976, Jimbo_1980}. This picture has also been recently reinforced by the work \cite{Galleas_2017} where we have shown that the spectrum 
of the six-vertex model's transfer matrix is governed by a Riccati equation. The latter is a non-linear differential equation and it appears in the six-vertex model as the differential equation underlying a particular functional relation obtained through the Algebraic-Functional (AF) framework.

Here, however, we find important to remark that the approach put forward in \cite{Galleas_2017} yields more than just the aforementioned differential equations. For instance, our approach provides a natural embedding for such equations along the lines of 
the \emph{Lax-pair} formulation of \emph{integrable evolution equations}. This analogy with Lax-pairs also enables a systematic construction of conserved quantities and, therefore, sheds some light on the relation between classical and quantum integrable systems.

\subsection{Riccati equations}
Linear first-order \emph{Ordinary Differential Equations} (ODEs) can be solved for generic coefficients. However, this scenario changes
drastically when we include non-linear terms. Among the most prominent non-linear ODEs we have the so called \emph{Riccati equation} \cite{Riccati_1724} which is a first-order ODE with a quadratic term. It is also one of the oldest non-linear differential equations considered in the literature and a large amount of works have been devoted to its study. In particular, we refer the reader to \cite{Reid_book}, \cite{Bittanti_book} and references therein for a more extensive discussion.

Riccati equation has made its appearance in several branches of theoretical physics and this relation was initially due to the special role played by Riccati equations in the theory of Bessel functions \cite{Watson_book}.
Several other connections with physical theories have been unveiled over the years and the relevance of Riccati equations  
has been steadily growing since its original proposal. For instance, it has also appeared in the \emph{factorization method} \cite{Schroedinger_1940, Schroedinger_1941a, Schroedinger_1941b, Infeld_Hull_1951}, in Witten's supersymmetric quantum mechanics \cite{Witten_1981} and in the theory of conformal mappings \cite{Ovsienko_Tabachnikov_book}. Moreover, Riccati equation shares a deep connection with Sturm-Liouville problems; in particular with \emph{Schr\"odinger equations} \cite{Haley_1997}. This latter relation has
been already exploited in \cite{Galleas_2017}; and it is responsible for unveiling Schr\"odinger equations solved by the six-vertex model's spectrum. 

The spectrum of the six-vertex model's transfer matrix has been thoroughly investigated in the literature over the years, mostly through Bethe ansatz methods, and a large amount of physical properties have been computed exactly to date. Nevertheless, very little is known about the model's \emph{continuous symmetries} and the investigation of such symmetries is the main goal of the present paper. In order
to clarify our goals, let us first elaborate on the case of \emph{discrete symmetries}.
In the case of two-dimensional vertex models an example of discrete symmetry is the so called \emph{crossing symmetry} as the latter does not depend on a continuous parameter. 
On the other hand, \emph{conformal symmetry} is an example of continuous symmetry for a statistical mechanical model and its existence
is not always easy to recognize. For instance, in the case of the two-dimensional Ising model, conformal symmetry was only proved in the early 2000s by Smirnov  \cite{Smirnov_2001} despite decades of accumulated evidences. 
This is the point where our Riccati representation of the six-vertex model's eigenvalue problem plays a differentiated role. As we are dealing with a differential equation from the start, we can use the \emph{Lie groups method} to unveil continuous symmetries of the six-vertex model in a systematic way. This is what we intend to show in the following sections.

\subsection{Lie groups method and continuous symmetries}

The modern concept of Lie algebras/groups has its roots in Sophus Lie's efforts to understand the variety of integration methods 
available for differential equations from the same principle. Lie's approach to that problem was largely influenced by \emph{Galois theory} of algebraic equations and it aims at finding solutions of a given differential equation through a careful examination of its symmetries. Here we refer to symmetries as a \emph{group transformation}, acting as a change of variables, such that the differential equation of interest remains invariant. These group symmetries will then map the solution space into itself and can often be used to reduce the order of the differential equation of interest. 
Moreover, sometimes it is also possible to use such group symmetries to bring the differential equation to a simpler form whose integration can be performed in a straightforward manner.  In this sense, Lie's work provides a classification of all differential equations for which \emph{integration} or \emph{lowering of order} can be affected by group theoretical methods \cite{Lie_1883, Lie_1884}.

\subsection{Integrability}
Riccati equations are also present in the theory of integrable systems. In particular, they appear within the context of isomonodromic deformations \cite{Babelon_book} and in the Bethe ansatz analysis of the Jaynes-Cummings-Gaudin model \cite{Babelon_Talalaev_2007}. However, some of their remarkable properties do not seem a priori consequences of integrability. For instance, Euler showed in the mid eighteenth-century that if a particular solution of the Riccati equation is known, then a general solution can be found using two quadratures. See for instance \cite{Boyce_DiPrima_book} for more details on Euler's formula. 
It is important to remark that Euler's formula also resembles the kind of relation obtained through the Lie groups method but we are unaware if such type of connection has ever been established in the literature. Despite the aforementioned similarity, the Lie groups method still provides a whole \emph{group} of maps from the solution space into itself, whilst Euler's formula consists of a single relation.
 
One can think of several ways to exploit Euler- and/or Lie-groups-type relations but both approaches still require an appropriate \emph{seed solution} in order to span the whole solution space. This is point where some difficulty can be expected since the possibility of writing down a particular solution of the Riccati equation still depends strongly on the particular form of its coefficients.
In this way, \emph{integrability} still plays an important role for Riccati equations and one might even expect 
an enhancement of the symmetries obtained through the Lie groups method. In particular, as for the Riccati equation obtained in
\cite{Galleas_2017}, we also have available a generating function of conserved quantities and one can expect our equations to exhibit additional properties when compared to generic Riccati equations.

\subsection{Outline}
In order to demonstrate how continuous symmetries of the six-vertex model can be studied using the Lie groups method we have organized this paper as follows. In 
\Secref{sec:AFM} we give a brief overview of the results previously obtained in \cite{Galleas_2017}. In fact, we limit the exposition of \Secref{sec:AFM} to the results
which will be relevant to the present paper. \Secref{sec:RIC} is then devoted to a generalization of the Riccati equations presented in \cite{Galleas_2017}. For instance, in our previous work we have presented explicit Riccati equations describing the six-vertex model's spectrum only for the $\mathfrak{h}$-modules with label $n=1$ and $n=2$. In \Secref{sec:RIC} we then generalize those results to all $\mathfrak{h}$-modules constituting the vector space where the transfer matrix acts on.
The analysis of the symmetry group associated to our Riccati equations is then presented in \Secref{sec:GEN}. In particular, in \Secref{sec:GEN} we present explicit
expressions for the symmetry generators in the cases $n=1, 2$ and examine their properties. Concluding remarks are then discussed in \Secref{sec:CONCL} and in the Appendices  \ref{app:LIE} through \ref{app:PHI} we give a description of the Lie groups method and present explicit expressions for functions appearing in the main text.

\section{Algebraic-Functional Method} \label{sec:AFM}

Functional equations methods play a fundamental role in most, if not all, exactly solvable models of Statistical Mechanics \cite{Baxter_book}. For instance, as far as the spectrum of the six-vertex model's transfer matrix is concerned, we have 
available at least three different types of functional equations which can be used to characterize the sought eigenvalues. In particular, 
here we highlight the so called \emph{\textsc{t-q} equations} \cite{Baxter_book}, \emph{inversion relation} \cite{Stroganov_1979} and the \emph{fusion hierarchy} \cite{Kulish_Reshetikhin_Sklyanin_1981, Kirillov_Reshetikhin_1987}.
Except for the inversion relation, all those functional methods require the introduction of extra objects in addition to
the transfer matrix one would like to diagonalize. On the other hand, although the method of inversion relations do not 
require additional objects, the derivation of such relations seems to be possible only in \emph{free-fermion} models.

The above mentioned methods are the most commonly used (functional) techniques to study vertex models possessing a commuting transfer matrix. However, alternative functional methods are also available. As a matter of fact, in the present paper we intend to investigate properties of the six-vertex model's spectrum which have been
only recently possibilitated by an alternative approach to that problem. We refer to the latter as \emph{Algebraic-Functional Method} (AFM)
and its basic idea is using the Yang-Baxter algebra as a source of functional equations characterizing quantities of interest. The AFM was put forward in \cite{Galleas_2008, Galleas_2015} and it went through several refinements before culminating in \cite{Galleas_2017}. In particular, in the work \cite{Galleas_2017} we have established an analogy between the AFM and the \emph{Classical Inverse Scattering Method} (CISM). Such analogy not only contains a systematic mechanism generating families of conserved quantities but also encompasses the main features of the CISM.

Within the AFM we regard the spectrum of the six-vertex model's transfer matrix as a family of curves embedded
in a non-linear functional/differential equation. The latter arises as the consistency condition of a linear functional equation derived through the AFM. In this way we find the basic ingredients to establish an 
AFM/CISM analogy: linear functional equations obtained from the Yang-Baxter algebra are then regarded as 
\emph{Auxiliary Linear Problems}. The scenario is then similar to what we find in the context of classical integrable systems and our linear functional equations play the same role as the Lax and/or zero-curvature representations of a non-linear problem.

\subsection{Auxiliar Linear Problem}

Let $L \in \Z_{\geq 1}$ be the lattice length and use $n \in \Z_{\geq 0}$ with $n \leq L$ to label the $\mathfrak{h}$-modules organizing the vector space $( \C^2  )^{\otimes L}$ where the six-vertex model's transfer matrix acts on. Here we shall also use $\mathfrak{S}_n$ to denote the symmetric group acting on $n$ variables $x_i \in \C$. In this way, the auxiliary linear problem described in \cite{Galleas_2017} reads
\[ \label{MF}
\sum_{i=0}^n \mathrm{M}_i \; \mathcal{F}_n (x_0, x_1 , \dots , \widehat{x_i} , \dots , x_n) = 0
\]
for a function $\mathcal{F}_n \in \C \llbracket x_1^{\pm 1} , x_2^{\pm 1} , \dots , x_n^{\pm 1} \rrbracket^{\mathfrak{S}_n}$ and $\mathrm{M}_i$ a meromorphic function on complex variables $x_0$, $x_1$, $\dots$, $x_n$. In fact, as already discussed in \cite{Galleas_2017}, Eq. \eqref{MF} encloses a system of functional equations since $\mathcal{F}_n$ is a symmetric function but not $\mathrm{M}_i$. 
Hence, by requiring all equations contained in \eqref{MF} to be compatible we must therefore ask the condition
\[ \label{CC}
\mathrm{det} \left( \mathrm{M}_{i,j}  \right)_{0 \leq i, j \leq n} = 0 
\]
to be fulfilled with 
\< \label{mij}
\mathrm{M}_{i,j} = \begin{cases}
\pi_{0, j} \mathrm{M}_j \,\qquad i = 0 \cr
\pi_{0, j} \mathrm{M}_0 \qquad i = j \cr
\pi_{0, j} \mathrm{M}_i \qquad \text{otherwise} 
\end{cases} \; .
\>
In \eqref{mij} we write $\pi_{i, j} \in \mathfrak{S}_{n+1}$ for the $2$-cycle acting as the permutation of variables $x_i$ and $x_j$.

The interpretation of Eq. \eqref{MF} as an auxiliary linear problem is based on the following observation. Suppose we are interested on a function $f = f(x)$ satisfying a non-linear equation. The latter can then be encoded in the linear problem
\eqref{MF} in case one can exhibit suitable coefficients $\mathrm{M}_i$, depending explicitly on $f(x_j)$, such that condition \eqref{CC} corresponds to the non-linear equation satisfied by $f$. In \cite{Galleas_2017} we have shown this
formulation is sound for the eigenvalues $\Lambda(x)$ of the six-vertex model's transfer matrix. In particular, as for the
six-vertex model the AFM also gives us the coefficients $\mathrm{M}_i$. They explicitly read
\< \label{MI}
&& \mathrm{M}_i \coloneqq \nonumber \\
&& \begin{cases}
\displaystyle \phi_1 \prod_{j=1}^n \frac{a(x_j - x_0)}{b(x_j - x_0)}  \; \lambda_{\mathcal{A}} (x_0) + \phi_2 \prod_{j=1}^n \frac{a(x_0 - x_j)}{b(x_0 - x_j)} \; \lambda_{\mathcal{D}} (x_0) - \Lambda(x_0) \qquad \qquad \; i = 0 \cr
\displaystyle \frac{c(x_0 - x_i)}{b(x_0 - x_i)} \left[ \phi_1 \prod_{\substack{j=1 \\ j \neq i}}^n \frac{a(x_j - x_i)}{b(x_j - x_i)}  \; \lambda_{\mathcal{A}} (x_i) - \phi_2 \prod_{\substack{j=1 \\ j \neq i}}^n \frac{a(x_i - x_j)}{b(x_i - x_j)} \; \lambda_{\mathcal{D}} (x_i) \right] \qquad \text{otherwise}
\end{cases}  \nonumber \\
\>
with parameters $\phi_1, \phi_2 \in \C$ characterizing deviations from strictly periodic boundary conditions.
Moreover, in \eqref{MI} we have written $\lambda_{\mathcal{A}} (x) \coloneqq \prod_{j=1}^L a(x - \mu_j)$ and $\lambda_{\mathcal{D}} (x) \coloneqq \prod_{j=1}^L b(x - \mu_j)$ for the highest-weight functions with inhomogeneity parameters $\mu_i \in \C$. 
The functions $a(x)$ and $b(x)$, together with $c(x)$, are in their turn the statistical weights associated to configurations of the six-vertex model. In order to describe such functions it is convenient to distinguish 
between two well known cases: the \emph{rational} and the \emph{trigonometric} model. We then take $a(x) = x + 1$, $b(x) = x$ and $c(x) = 1$ when referring to the rational case and $a(x) = \sinh{(x + \gamma)}$, $b(x) = \sinh{(x)}$ and $c(x) = \sinh{(\gamma)}$ for the trigonometric model. In the latter case $\gamma \in \C$ is usually called anisotropy parameter. 
In fact, the rational model can be obtained from the trigonometric one in a particular limit but we prefer to keep in mind the two possibilities of statistical weights and declare when appropriate the model we are considering.

\section{Riccati equations} \label{sec:RIC}

Although there exist several methods for studying functional relations, it is fair to say the theory of \emph{functional equations} is not as well developed as the theory of \emph{differential equations}. 
In fact, some powerful methods for solving functional equations such as the \emph{derivative method} take our attention
away from the functional equation of interest and puts it on an associated differential equation.
As for the functional relation obtained from \eqref{CC} and \eqref{mij}, using coefficients $\mathrm{M}_i$ given by \eqref{MI}, we have also discussed in our previous work \cite{Galleas_2017} the differential equations underlying the 
cases $n=1$ and $n=2$. They turn out to be \emph{Riccati equations} and in this section we intend to elaborate on the analysis presented in \cite{Galleas_2017}; and also present Riccati equations governing the eigenvalues $\Lambda(x)$ for arbitrary values of $n$. Moreover, it is important to remark that in this section we will be considering the trigonometric six-vertex model, as discussed in \Secref{sec:AFM}, keeping in mind that equations for the rational model can be obtained in a particular limit.

\begin{prop} \label{ric_thm}
Let $l,m \in \{ 1, 2, \dots , n-1 \}$ and assume $\exists \; u_l , u_m \in \C \colon \Lambda(u_l) = \Lambda(u_m) = 0$ and $u_l \neq u_m + 2 \ii \pi \Z$.
Then solutions of the functional equation \eqref{CC} satisfy the Riccati equation
\[
\label{ric_gen}
\bar{\Omega}(x) \partial \Lambda(x) = \Omega_0 (x) - \Omega_1 (x) \Lambda(x) + \Omega_2 (x) \Lambda(x)^2
\]
with coefficients given by
\begin{align} \label{omegas}
\bar{\Omega}(x) &\coloneqq \left. \mathrm{det} \left( \omega_{i,j}  \right)_{\substack{0 \leq i, j \leq n \\ i, j \neq 1}} \right|_{x_0 = x_1 = x} & \Omega_0 (x) & \coloneqq  \partial_1 \left[ \frac{\mathrm{det} \left( \omega_{i,j}  \right)_{0 \leq i, j \leq n }}{b(x_1 - x_0)}  \right]_{x_0 = x_1 = x} \nonumber \\
\Omega_1 (x) & \coloneqq  \partial_1 \left[  \mathrm{det} \left( \omega_{i,j}  \right)_{\substack{0 \leq i, j \leq n \\ i, j \neq 0}} + \mathrm{det} \left( \omega_{i,j}  \right)_{\substack{0 \leq i, j \leq n \\ i, j \neq 1}} \right]_{x_0 = x_1 = x} & \Omega_2 (x) & \coloneqq \left. \mathrm{det} \left( \omega_{i,j}  \right)_{\substack{0 \leq i, j \leq n \\ i, j \neq 0, 1}} \right|_{x_0 = x_1 = x} \nonumber \\
\end{align}
and matrix entries $\omega_{i,j}$  defined in \Appref{app:omega}.
\end{prop}

\begin{proof}
We first look at the specialization $x_i = u_{i-1} \; \forall i \in \{ 2, 3, \dots, n \}$ of Eq. \eqref{CC}. 
Then we take the limit $x_0, x_1 \to x$ of the resulting equation to find \eqref{ric_gen} with coefficients \eqref{omegas}.
\end{proof}

\begin{rmk}
In order to clarify our notation we remark that in \eqref{ric_gen} and \eqref{omegas} we write
$\partial \coloneqq \frac{\dd}{\dd x}$ and $\partial_i \coloneqq \frac{\dd}{\dd x_i}$ to denote derivatives.
\end{rmk}

\begin{rema}
From the expressions given in \Appref{app:omega} we can readily see that the matrix entries $\omega_{i,j}$ 
do not exhibit poles when $x_0 \to x_1$.
\end{rema}

\begin{lem} \label{DB}
The functional dependence of $\mathrm{det} \left( \omega_{i,j}  \right)_{0 \leq i, j \leq n }$ on $x_0$ and $x_1$ is 
of the form $b(x_1 - x_0) \frac{\mathcal{R}(x_0,x_1)}{\mathcal{S}(x_0,x_1)}$ with both $\mathcal{R}$ and $\mathcal{S}$ trigonometric polynomials in $x_0$ and $x_1$. That is, polynomials in $e^{\pm x_0}$ and $e^{\pm x_1}$.
\end{lem}

\begin{proof}
From the explicit expressions given in \Appref{app:omega} one can readily notice that the coefficients $\omega_{i,j}$
can be written as $\omega_{i,j} = \kappa^{(1)}_{i,j} \left(\kappa^{(2)}_{i,j} \right)^{-1}$ with $\kappa^{(a)}_{i,j}$ $(a = 1,2)$ a trigonometric polynomial. Now let $\mathcal{M}$ be a $(n+1)\times (n+1)$ matrix with entries $\omega_{i,j}$ ($0 \leq i,j \leq n$) and write $\widetilde{\mathcal{M}}$ for an analogous matrix with coefficients $\widetilde{\omega}_{i,j}$ defined as
\[
\widetilde{\omega}_{i,j} \coloneqq \begin{cases}
- \omega_{i,j} \qquad  i = 1 \cr
\omega_{i,j} \qquad \text{otherwise}
\end{cases} \; .
\]
Clearly, $\mathrm{det} \left( \widetilde{\mathcal{M}} \right) = - \mathrm{det} \left( \mathcal{M} \right)$. Next notice $\widetilde{\omega}_{i,0} =\widetilde{\omega}_{i,1}$ when $x_0 = x_1$ which implies $\mathrm{det} \left( \left. \widetilde{\mathcal{M}} \; \right|_{x_0 = x_1} \right) = 0$ 
since two columns are equal. Hence, due to the (trigonometric) polynomial structure one can conclude 
\[
\mathrm{det} \left( \mathcal{M} \right) = b(x_1 - x_0) \frac{\mathcal{R}(x_0,x_1)}{\mathcal{S}(x_0,x_1)} \; .
\]
\end{proof}
\begin{rema}
The computation of $\Omega_0$ defined in \eqref{omegas} simplifies to 
\[
\Omega_0 (x) = \partial_1 \left( \frac{\mathcal{R}(x_0,x_1)}{\mathcal{S}(x_0,x_1)}  \right)_{x_0 = x_1 = x} 
\]
due to Lemma \ref{DB}. Hence, $\Omega_0$ does not contain poles when $x_0 = x_1$.    
\end{rema}    

We then proceed with our analysis by inspecting particular cases of Proposition \ref{ric_thm}.    
    
\subsection{Case $n=1$}
This is the first non-trivial case covered by Proposition \ref{ric_thm}. The coefficients $\bar{\Omega}$, $\Omega_0$, $\Omega_1$ and $\Omega_2$ in that case explicitly read
\< \label{coeff_n1}
\bar{\Omega} (x) &=& - c \lambda_{-} (x) \nonumber \\
\Omega_0 (x) &=&  \left[ \cosh{(\gamma)} \lambda_{+} (x)  \right]^2 - \left[ \sinh{(\gamma)} \lambda_{-} (x)  \right]^2  +  \sinh{(\gamma)} \cosh{(\gamma)} \left[ \lambda_{+}(x) \partial \lambda_{-}(x) - \lambda_{-}(x) \partial \lambda_{+}(x)   \right] \nonumber \\
\Omega_1 (x) &=& 2 \cosh{(\gamma)} \lambda_{+}(x) + \sinh{(\gamma)} \partial \lambda_{-} (x) \nonumber \\
\Omega_2 (x) &=& 1 \; . \nonumber \\
\>

\begin{rema}
Eq. \eqref{ric_gen} with coefficients \eqref{coeff_n1} has been already discussed in \cite{Galleas_2017}; although derived from a different procedure.
\end{rema}

\subsection{Case $n=2$}
We carry on to the case $n=2$ and it is important to remark that this is the first instance of Proposition \ref{ric_thm} with coefficients \eqref{omegas} depending on parameters $u_k$.
In that case we have the following coefficients for the Riccati equation \eqref{ric_gen},
\< \label{coeff_n2}
\bar{\Omega} (x) &=&  \frac{\sinh{(\gamma)}^2}{4 \sinh{(u_1 - x)}^2} \left[ (l_{+} + l_{-}) \sinh{[2(u_1 - x + \gamma)]} + (l_{-} - l_{+}) \sinh{[2(x - u_1 + \gamma)]} \right] \lambda_{+}(x) \nonumber \\
&& - \; \frac{\sinh{( 2 \gamma)}}{4 \sinh{(u_1 - x)}^2} \left[ (l_{+} + l_{-}) \sinh{(u_1 - x + \gamma)}^2 + (l_{+} - l_{-}) \sinh{(x - u_1 + \gamma)}^2 \right] \lambda_{-}(x) \nonumber \\
\Omega_0 (x) &=& 2 l_{+} \left[ \sinh{(\gamma)} \cosh{(\gamma)}  \frac{\lambda_{-}(x)}{\sinh{(u_1 - x)}}  \right]^2 \nonumber \\ 
&& + \; \frac{\lambda_{+}(x)}{\sinh{(u_1 - x)}} \left[ \cosh{(2 \gamma)} \lambda_{+}(x) + \sinh{(\gamma)} \cosh{(\gamma)} \partial \lambda_{-}(x)  \right] \mathcal{P}(u_1 - x) \nonumber \\
&& - \; \sinh{(\gamma)} \cosh{(\gamma)}  \frac{\lambda_{-}(x)}{\sinh{(u_1 - x)}} \left[ \mathcal{P}(2(u_1 - x)) \frac{\lambda_{+}(x)}{\sinh{(u_1 - x)}}  + \mathcal{P}(u_1 - x) \partial \lambda_{+}(x)   \right] \nonumber \\
\Omega_1 (x) &=& - \frac{\mathcal{Q}(x)}{2 \sinh{(u_1 - x)}^2} \left[ \sinh{(2 \gamma)} \lambda_{-}(x) + \sinh{(\gamma)}^2 \partial \lambda_{+}(x)  \right] \nonumber \\
&& + \; \frac{\lambda_{+}(x)}{\sinh{(u_1 - x)}^2} \left[ l_{+} \cosh{(2 \gamma)} \left( \cosh{(\gamma)}^2 \cosh{(2(u_1 - x))} - 1 \right) \right. \nonumber \\
&& \qquad \qquad \qquad\qquad\quad + \; \left.  l_{-} \sinh{(2 \gamma)} \cosh{(\gamma)}^2 \sinh{(2(u_1 - x))} \right] \nonumber \\
&& + \; \frac{\sinh{(2 \gamma)} \partial \lambda_{-}(x) }{4 \sinh{(u_1 - x)}^2} \left[ l_{+} \left( \sinh{(u_1 - x + \gamma)}^2 + \sinh{(x - u_1  + \gamma)}^2 \right) \right. \nonumber \\
&& \qquad \qquad\qquad\qquad\; \qquad\qquad \;\;\;\quad\qquad \qquad + \left. l_{-} \sinh{(2 \gamma)} \sinh{(2(u_1 - x))}   \right] \nonumber \\
\Omega_2 (x) &=&  \frac{1}{2 \sinh{(u_1 - x)}^2} \left[ (l_{+} + l_{-}) \sinh{(u_1 - x + \gamma)}^2 + (l_{+} - l_{-}) \sinh{(x - u_1 +\gamma)}^2   \right] \; . \nonumber \\
\>
As for the notation employed in \eqref{coeff_n2} we write $l_{\pm} \coloneqq \lambda_{\pm} (u_1)$ and 
\<
\mathcal{P}(x) &\coloneqq& l_{+} \cosh{(2 \gamma)} \sinh{(x)} + l_{-} \sinh{(2 \gamma)} \cosh{(x)} \nonumber \\
\mathcal{Q}(x) &\coloneqq& l_{+} \cosh{(2 \gamma)} \sinh{(2(u_1 - x))} + l_{-} \sinh{(2 \gamma)} \cosh{(2(u_1 - x))} \; .    
\>

Some comments are in order at this stage. In \cite{Galleas_2017} we have already presented a Riccati equation describing the eigenvalues $\Lambda$ in the $\mathfrak{h}$-module corresponding to $n=2$; however, not only the equation obtained in \cite{Galleas_2017} differs from \eqref{ric_gen} with coefficients \eqref{coeff_n2} but also its derivation. 
For instance, in \cite{Galleas_2017} we have resorted to conserved quantities in order to unveil such type of equation.
Here, however, we have obtained \eqref{coeff_n2} through a simpler procedure which can also be used to derive the same equation of \cite{Galleas_2017}. That is what we intend to demonstrate next. 

As for this alternative derivation we first look at the limit $x_0 , x_1 \to x$ of \eqref{CC}. This is essentially the same procedure used to find \eqref{ric_gen} but we now choose a different specialization for the variable $x_2$. 
More precisely, we set $x_2 = 0$ and by doing so we find 
\[ \label{riccati2}
\bar{J}(x) \partial \Lambda(x) + K_{+} (x) \Lambda(x)^2 + J_1 (x) \Lambda(x) + J_0 (x) = 0 
\]
with coefficients 
\< \label{coeff}
\bar{J}(x) &\coloneqq& \lambda_{+}(x) \sinh{(\gamma)}^2 K_{-} (x) + \; \lambda_{-}(x) \sinh{(\gamma)} \cosh{(\gamma)} K_{+} (x) \nonumber \\
J_1 (x) &\coloneqq& K_{-} (x) \left[ \sinh{(2 \gamma)} \lambda_{-}(x) + \sinh{(\gamma)}^2 \partial \lambda_{+} (x)    \right] - \frac{1}{2} K_{+} (x) \partial\lambda_{-} (x) \nonumber \\
&& + \; 2 \lambda_{+} (x) \left\{ \cosh{(2 \gamma)} \left[ m_{+} \left(1 - \cosh{(\gamma)}^2 \cosh{(2 x)} \right) - \Lambda_0  \right] \right. \nonumber \\
&& \qquad \qquad \;\; + \; \left. \cosh{(\gamma)}^2 \left[ \Lambda_0 \cosh{(2 x)} + m_{-} \sinh{(2 \gamma)} \sinh{(2 x)}  \right] \right\} \nonumber \\
J_0 (x) &\coloneqq& (m_{+} - \Lambda_0) \sinh{(2 \gamma)}^2 \lambda_{-}(x)^2 + \left[ \cosh{(2 \gamma)} \lambda_{+}(x) + \sinh{(\gamma)} \cosh{(\gamma)} \partial \lambda_{-} (x)  \right] \lambda_{+} (x) K_0 (x) \nonumber \\
&& + \; \sinh{(\gamma)} \cosh{(\gamma)} \lambda_{-}(x) \left[ 2 \lambda_{+}(x) K_{-} (x) - K_0 (x) \partial \lambda_{+}(x)  \right] \; . \nonumber \\
\>
The coefficients $J_0$, $J_1$ and $\bar{J}$ are written down in terms of functions $K_{\pm}$ and $K_{0}$; which are in their turn defined as
\< \label{coeff_aux}
K_{+} (x) &\coloneqq& m_{+} ( \cosh{(2 x)} \cosh{(2\gamma)} - 1) - m_{-} \sinh{(2 \gamma)} \sinh{(2 x)} - 2 \Lambda_0 \sinh{(x)}^2  \nonumber \\
K_{-} (x) &\coloneqq& m_{+}  \sinh{(2 x)} \cosh{(2\gamma)}  - m_{-} \sinh{(2 \gamma)} \cosh{(2 x)} -  \Lambda_0 \sinh{(2 x)}  \nonumber \\
K_0 (x) &\coloneqq& 2 m_{+} \cosh{(2 \gamma)} \sinh{(x)}^2 - m_{-} \sinh{(2 \gamma)} \sinh{(2 x)} + \Lambda_0 \left( \cosh{(2 \gamma)} - \cosh{(2 x)} \right) \; . \nonumber \\
\>
In \eqref{coeff} and \eqref{coeff_aux} we have also employed additional conventions, namely $m_{\pm} \coloneqq \lambda_{\pm} (0)$ and $\Lambda_0 \coloneqq \Lambda(0)$.

\begin{rema}
Eq. \eqref{riccati2} with coefficients given by \eqref{coeff} and \eqref{coeff_aux} is in fact a generalization of the equation presented in \cite{Galleas_2017}. More precisely, \eqref{riccati2} is valid for generic parameters $\phi_1, \phi_2 \in \C^{\times}$ and $\mu_j \in \C$; whilst the 
equation of \cite{Galleas_2017} holds only for $\phi_1 = \phi_2 = 1$ and $\mu_j = 0$.
\end{rema}

\begin{rema}
One particular aspect of \eqref{coeff} and \eqref{coeff_aux} is the presence of the initial condition $\Lambda_0 = \Lambda(0)$. Although the latter can be readily evaluated when $\mu_j = 0$, see for instance \cite{Galleas_2008, Galleas_2017}, a similar computation does not seem doable for generic values of the inhomogeneity parameters $\mu_j$.
\end{rema}

\subsection{Uniqueness} \label{sec:UNI}
Riccati equations usually exhibit non-unique solutions. This can be readily seen from Euler's formula \cite{Boyce_DiPrima_book} which
extends a particular solution to a generic one-parameter dependent solution. Here, however, we are not dealing with a generic Riccati equation but with one having the particular coefficients \eqref{omegas}. The latter depends on $n-1$ parameters $u_k$ which are, by construction, distinct zeroes of the sought eigenvalue $\Lambda$. More precisely, \eqref{ric_gen} is defined upon the condition
$\Lambda (u_k) = 0$ and, as long as we are looking for solutions having more than $n-1$ distinct zeroes, there is still room for non-unique solutions. 

In addition to that it is important to keep in mind that we would like to use \eqref{ric_gen} with \eqref{omegas} to describe eigenvalues of the six-vertex model's transfer matrix and such eigenvalues live in a subset of the solution space satisfying additional properties. For instance, they consist of (trigonometric) polynomials with degree higher than $n-1$ in most cases.
Hence, such solutions $\Lambda$ are expected to have more than $n-1$ zeroes and the following arguments 
can be used to partially address the question of uniqueness of solutions:
\begin{enumerate}[label=\emph{\roman*})]
\item Solutions of \eqref{ric_gen} with coefficients \eqref{omegas} have at least $n-1$ distinct zeroes. If this is not the case they can not satisfy the intended Riccati equation since it is defined only for functions $\Lambda$ exhibiting a minimum number of $n-1$ distinct zeroes.
\\
\item (Trigonometric) polynomials solutions with precisely $n-1$ distinct zeroes, in case they exist, 
are unique up to an overall multiplicative factor. 
This is justified by the following argument. Let $\Lambda_1$ be a polynomial with $n-1$ zeroes $\{ u_k^{(1)} \}$ and similarly write $\Lambda_2$ for a polynomial with $n-1$ zeroes $\{ u_k^{(2)} \}$. If $\Lambda_1$ satisfies \eqref{ric_gen} with \eqref{omegas} we necessarily have $\{ u_k \} =\{ u_k^{(1)} \}$. Hence, if $\Lambda_2$ also solves the same equation we can likewise conclude $\{ u_k \} =\{ u_k^{(2)} \}$. Consequently, $\{ u_k^{(1)} \} = \{ u_k^{(2)} \}$ which implies that
$\Lambda_1$ and $\Lambda_2$ can only differ by an overall factor.
\end{enumerate}

In conclusion, the above arguments still leave space for non-unique solutions in a space of functions possessing more than $n-1$ zeroes.
Also, in case non-unique solutions exist, two solutions having more than $n-1$ zeroes would necessarily share $n-1$ zeroes by the same argument. 
On the other hand, the explicit diagonalization of the transfer matrix for small values of $n$ and $L$ reveals that there are no two eigenvalues $\Lambda$ in the same $\mathfrak{h}$-module sharing one single zero. This observation then suggests 
that polynomial solutions of \eqref{ric_gen} with coefficients \eqref{omegas} might indeed be unique for each set of 
$n-1$ zeroes $\{ u_k \}$.

\subsection{Fixing the zeroes $u_k$}

In the previous subsection we have elaborated on the \emph{uniqueness problem} associated to the Riccati equation \eqref{ric_gen} with coefficients \eqref{omegas}. As for that discussion we have based our arguments on the special role played by the zeroes $u_k$ in the explicit formulae for the coefficients \eqref{omegas}. In particular, our Riccati equation will depend on $n-1$ complex parameters $u_k$ which have not been determined as yet. 
We shall restrict the discussion of this subsection to the transfer matrix's eigenvalue problem; and for that we also need to impose the
\emph{boundary conditions} described in \cite{Galleas_2017}. The latter are meant to filter solutions of \eqref{ric_gen} belonging to the spectrum of the six-vertex model's transfer matrix. In this way, we look for solutions which are trigonometric polynomials of degree $L$ and without lack of generality one
can write such eigenvalues $\Lambda$ as
\[ \label{LL}
\Lambda (x) = \Lambda_0 \prod_{j=1}^L \frac{\sinh{(u_j - x)}}{\sinh{(u_j)}} \; .
\]
It is worth remarking that such type of representation appeared previously in \cite{Galleas_2008}. 
In order to proceed it is convenient to introduce some extra conventions. For instance, we shall write
\[
\mathcal{U}_{\Lambda} \coloneqq \{ u_j \in \C \mid \Lambda (u_j) = \Lambda (u_{j'}) =  0 \; , \; u_j \neq u_{j'} + 2 \ii \pi \Z \; , \; 1 \leq j, j' \leq L \}
\]
for the set of all zeroes $u_k$ of a particular eigenvalue $\Lambda$. Also, let $\mathcal{I}_m$ be a subset of $\mathcal{U}_{\Lambda}$ of cardinality $n-1$ and define $\mathcal{I} \coloneqq \{ \mathcal{I}_m \}$ as the set composed of all possible $\frac{L!}{(n-1)! (L-n+1)!}$ subsets $\mathcal{I}_m$. Clearly $\mathcal{U}_{\Lambda}  = \bigcup_m \mathcal{I}_m$.
In addition to that we extend the notation employed in \eqref{ric_gen} and \eqref{omegas} to $\bar{\Omega} = \bar{\Omega} (x , \mathcal{I}_m)$, $\Omega_0 = \Omega_0 (x , \mathcal{I}_m)$, $\Omega_1 = \Omega_1 (x , \mathcal{I}_m)$ and $\Omega_2 = \Omega_2 (x , \mathcal{I}_m)$; in order to emphasize the dependence of such coefficients on the $n - 1$ particular 
zeroes $u_k \in \mathcal{I}_m $.

Next we notice representation \eqref{LL} implies the relation $\partial \Lambda(x) = - \Lambda(x) \mathscr{F}(x, \mathcal{U}_{\Lambda})$
with 
\[ \label{FU}
\mathscr{F}(x, \mathcal{U}_{\Lambda}) \coloneqq \sum_{u \in \mathcal{U}_{\Lambda}} \coth{(u - x)} \; .
\]
Therefore, representation \eqref{LL} converts our Riccati equation \eqref{ric_gen} into a quadratic algebraic equation, namely
\[ \label{quad}
\Omega_2 (x , \mathcal{I}_m) \Lambda(x)^2 + \left[ \bar{\Omega} (x , \mathcal{I}_m) \mathscr{F}(x, \mathcal{U}_{\Lambda}) - \Omega_1 (x , \mathcal{I}_m)  \right] \Lambda(x) + \Omega_0 (x , \mathcal{I}_m) = 0 \; .
\]
As for the set $\mathcal{U}_{\Lambda}$ we then find the following relations fixing it. 

\begin{lem} \label{ZEROESU}
Let $\mathrm{Sym}(\mathcal{U}_{\Lambda})$ denote the group of permutations on the set $\mathcal{U}_{\Lambda}$ and further write
$\mathrm{\Pi}_{\sigma} \colon f(u_1 , u_2 , \dots , u_L) \mapsto f(\sigma(u_1) , \sigma(u_2) , \dots , \sigma(u_L))$ for functions $f$ on
$\mathcal{U}_{\Lambda}$ and $\sigma = \sigma(u_1) \sigma(u_2) \dots \sigma(u_L) \in  \mathrm{Sym}(\mathcal{U}_{\Lambda})$ a particular permutation. Also, write
\[
\Delta \coloneqq \begin{pmatrix} 
\Omega_0 (0 , \mathcal{I}_m) & \Omega_1 (0 , \mathcal{I}_m) & \bar{\Omega} (0 , \mathcal{I}_m) & \Omega_2 (0 , \mathcal{I}_m) \\
\Omega_0 (0 , \mathcal{I}_{\bar{m}}) & \Omega_1 (0 , \mathcal{I}_{\bar{m}}) & \bar{\Omega} (0 , \mathcal{I}_{\bar{m}}) & \Omega_2 (0 , \mathcal{I}_{\bar{m}})
\end{pmatrix}
\]
and $\Delta_{I,J} (\mathcal{I}_m, \mathcal{I}_{\bar{m}})$ for the determinant of the matrix obtained from $\Delta$ by removing columns 
$I$ and $J$. The zeroes $u_k$ then satisfy the system of equations formed by
\< \label{ZU}
\left( \phi_1 e^{(L-2)\gamma} + e^{2\gamma} \phi_2  \right) \prod_{j=1}^L  e^{u_j - \mu_j}  \sinh{(u_j)} = \frac{ (-1)^L \mathrm{\Pi}_{\sigma} (\Delta_{2,3} (\mathcal{I}_m, \mathcal{I}_{\bar{m}}))}{\mathrm{\Pi}_{\sigma}(\Delta_{1,3} (\mathcal{I}_m, \mathcal{I}_{\bar{m}})) - \mathrm{\Pi}_{\sigma}(\Delta_{1,2} (\mathcal{I}_m, \mathcal{I}_{\bar{m}})) \sum_{j=1}^L \coth{(u_j)}} \nonumber \\
\hfill \sigma \in \mathrm{Sym}(\mathcal{U}_{\Lambda}); \; \mathcal{I}_m, \mathcal{I}_{\bar{m}} \in \mathcal{I} 
\>
and 
\[ \label{ph12}
\frac{\phi_1}{\phi_2} = - \frac{\sinh{\left( 2\gamma + \sum_{j=1}^L (u_j - \mu_j) \right)}}{\sinh{\left( (L-2)\gamma + \sum_{j=1}^L (u_j - \mu_j) \right)}} \; .
\]
\end{lem}
\begin{proof}
We first remark Eq. \eqref{quad} depends on an arbitrary subset $\mathcal{I}_m \subset \mathcal{U}_{\Lambda}$ of cardinality $n-1$.
Now, since $\mathrm{card}(\mathcal{U}_{\Lambda}) = L$ and $n \leq L$, Eq. \eqref{quad} also holds with $\mathcal{I}_m \mapsto \mathcal{I}_{\bar{m}}$ where $\mathcal{I}_{\bar{m}}$ is another subset of $\mathcal{U}_{\Lambda}$ such that $\mathcal{I}_{\bar{m}} \neq \mathcal{I}_m$. In other words, there exist several quadratic equations of the form \eqref{quad} describing the same eigenvalue $\Lambda$. We then use a pair of them, i.e. one with $u_k \in \mathcal{I}_m$ and another with $u_k \in \mathcal{I}_{\bar{m}}$, to eliminate the quadratic term $\Lambda(x)^2$. By doing so we are left with the expression
\[ \label{LX}
\Lambda(x) = \frac{\Delta_{2,3} (x \mid \mathcal{I}_{m}, \mathcal{I}_{\bar{m}})}{\Delta_{1,3} (x \mid \mathcal{I}_{m}, \mathcal{I}_{\bar{m}}) - \mathscr{F}(x, \mathcal{U}_{\Lambda}) \Delta_{1,2} (x \mid \mathcal{I}_{m}, \mathcal{I}_{\bar{m}})}
\]
where
\<
\Delta_{1,2} (x \mid \mathcal{I}_{m}, \mathcal{I}_{\bar{m}}) &\coloneqq& \bar{\Omega} (x, \mathcal{I}_{m}) \Omega_2 (x, \mathcal{I}_{\bar{m}})  - \bar{\Omega} (x, \mathcal{I}_{\bar{m}}) \Omega_2 (x, \mathcal{I}_{m}) \nonumber \\
\Delta_{1,3} (x \mid \mathcal{I}_{m}, \mathcal{I}_{\bar{m}}) &\coloneqq& \Omega_1 (x, \mathcal{I}_{m}) \Omega_2 (x, \mathcal{I}_{\bar{m}})  - \Omega_1 (x, \mathcal{I}_{\bar{m}}) \Omega_2 (x, \mathcal{I}_{m}) \nonumber \\
\Delta_{2,3} (x \mid \mathcal{I}_{m}, \mathcal{I}_{\bar{m}}) &\coloneqq& \Omega_0 (x, \mathcal{I}_{m}) \Omega_2 (x, \mathcal{I}_{\bar{m}})  - \Omega_0 (x, \mathcal{I}_{\bar{m}}) \Omega_2 (x, \mathcal{I}_{m}) \; .
\>
Formula \eqref{LX} must now be consistent with the representation \eqref{LL} and one can find enough equations characterizing the zeroes $u_k$ by identifying \eqref{LX} and \eqref{LL}. However, this naive approach renders equations which are hard to write down in closed form. On the other hand, we can also substitute \eqref{LL} directly in \eqref{quad} and then inspect the limits $x \rightarrow \pm \infty$ which are doable. As a result we find the conditions
\[ \label{L0}
\Lambda_0 =  (-1)^L \left( \phi_1 e^{(L-2)\gamma} + e^{2\gamma} \phi_2  \right) \prod_{j=1}^L  e^{u_j - \mu_j}  \sinh{(u_j)} 
\]
and
\[ \label{PH12}
\frac{\phi_1}{\phi_2} = - \frac{\sinh{\left( 2\gamma + \sum_{j=1}^L (u_j - \mu_j) \right)}}{\sinh{\left( (L-2)\gamma + \sum_{j=1}^L (u_j - \mu_j) \right)}} \; .
\]
Moreover, we further notice the property $\mathrm{\Pi}_{\sigma} \left( \Lambda(x) \right) = \Lambda(x)$ is explicitly manifested in representation \eqref{LL} but not in \eqref{LX}. Equations \eqref{ZU} then follow from the identification $\Lambda_0 =  \left. \mathrm{\Pi}_{\sigma} \left( \Lambda(x) \right) \right|_{x = 0}$ with $\Lambda_0$ given by \eqref{L0} and $\Lambda(x)$ obtained from
\eqref{LX}. The last equation, namely \eqref{ph12}, corresponds to the condition \eqref{PH12} obtained from the asymptotic analysis of \eqref{quad}. This completes our proof.
\end{proof}

\begin{rema}
Eq. \eqref{ph12} is not well defined for the particular case $L=4$. In that case the same procedure used in the proof of Lemma \ref{ZEROESU} yields the condition
\[
e^{4 \gamma - 2 \sum_{j=1}^4 (u_j - \mu_j)} = 1 \nonumber 
\]
instead of \eqref{ph12}.
\end{rema}

\begin{rema}
The system of equations described in  Lemma \ref{ZEROESU} is certainly not the only possible way to characterize the zeroes
$u_k$ using the quadratic equation \eqref{quad}. However, in the present paper we shall refrain ourselves from examining other possibilities.
\end{rema}

\section{Symmetry algebra generators} \label{sec:GEN}

In this section we discuss the implementation of the Lie groups method for the Riccati equation presented in \Secref{sec:RIC}. For that one first need to characterize the surface $\Sigma$ and subsequently construct the prolonged vector field $\bm{v}_l$ with $l$  the order of the differential equation of interest. We refer the reader to \Appref{app:LIE} for a short summary of this method.

The surface $\Sigma$ describing \eqref{ric_gen} reads
\[
\Sigma = \bar{\Omega}(x) \Lambda^{(1)} - \Omega_0 (x) + \Omega_1 (x) \Lambda - \Omega_2 (x) \Lambda^2  
\]
and, as \eqref{ric_gen} is a first-order differential equation, we have $l=1$. 
In that case we also need to write down explicitly the prolonged vector field $\bm{v}_1$. The latter is given by
\[ \label{bmv1}
\bm{v}_1 = \xi(x, \Lambda) \frac{\partial}{\partial x} + \phi (x, \Lambda) \frac{\partial}{\partial \Lambda} + \left[ \phi_x + (\phi_{\Lambda} - \xi_{x}) \Lambda^{(1)} - \xi_{\Lambda} \left( \Lambda^{(1)} \right)^2 \right] \frac{\partial}{\partial\Lambda^{(1)}} \; .
\]
The functions $\xi$ and $\phi$ appearing in \eqref{bmv1} are in their turn determined from the simultaneous resolution of
\[ \label{V1}
\bm{v}_1 \; \Sigma = 0 \qquad \text{and} \qquad \Sigma = 0 \; .
\]
Linearly independent solutions of \eqref{V1} are then used to build vector fields 
$\bm{v}_0 = \xi(x, \Lambda) \frac{\partial}{\partial x} + \phi (x, \Lambda) \frac{\partial}{\partial \Lambda}$ 
which will generate the symmetry algebra underlying the differential equation \eqref{ric_gen}.

Now we turn our attention to the resolution of the system of equations \eqref{V1}. 
For that we shall consider the \emph{minimal ansatz} \footnote{We refer to \eqref{ansatz} as minimal ansatz since they are the polynomials in $\Lambda$ of least degree for which we find non-trivial solutions for their $x$-dependent coefficients.}
\< \label{ansatz}
\xi(x, \Lambda) &=& f_0 (x) \nonumber \\
\phi(x, \Lambda) &=& g_0 (x) + g_1 (x) \Lambda  
\>
whose substitution in \eqref{V1} yields the following system of equations characterizing the functions $f_0$, $g_0$ and $g_1$:
\< \label{deter}
\left[ f_0^{\prime} (x)  + g_1(x) \right] \Omega_2(x) \bar{\Omega}(x) + f_0(x) \left[ \Omega_2^{\prime}(x) \bar{\Omega}(x) - \Omega_2(x) \bar{\Omega}^{\prime}(x)  \right] &=& 0 \nonumber \\
 \left[ f_0^{\prime} (x) \Omega_1(x) + g_1^{\prime} (x) \bar{\Omega}(x)  \right] \bar{\Omega}(x) \;\;+ f_0(x) \left[ \Omega_1^{\prime}(x) \bar{\Omega}(x) - \Omega_1(x) \bar{\Omega}^{\prime}(x)  \right] \nonumber \\
- \; 2 g_0(x) \Omega_2(x) \bar{\Omega}(x) &=& 0 \nonumber \\
\left[ f_0^{\prime} (x)  - g_1(x) \right] \Omega_0(x) \bar{\Omega}(x) - \left[ g_0 (x) \Omega_1(x) + g_0^{\prime} (x) \bar{\Omega}(x)  \right] \bar{\Omega}(x) \nonumber \\
+ \; f_0(x) \left[ \Omega_0^{\prime}(x) \bar{\Omega}(x) - \Omega_0(x) \bar{\Omega}^{\prime}(x)  \right] &=& 0 \; .
\>

Hence, the use of the ansatz \eqref{ansatz} for solving \eqref{V1} leaves us with a system of three differential equations for the aforementioned three functions. It is also worth remarking that in \eqref{deter} we have employed the symbol $^{\prime}$ to denote differentiation with respect to the variable $x$.

Solutions of \eqref{deter} can be obtained as follows. We firstly use the first equation of \eqref{deter} to express the function $g_1$ in terms of $f_0$ and $f_0^{\prime}$. By doing so we readily obtain
\[ \label{G1}
g_1 (x) =  f_0(x) \frac{\left[ \Omega_2(x) \bar{\Omega}^{\prime}(x) - \Omega_2^{\prime}(x) \bar{\Omega}(x)  \right]}{\Omega_2(x) \bar{\Omega}(x)} - f_0^{\prime} (x) \; .
\]
Next we substitute \eqref{G1} in the second equation of \eqref{deter} to find
\< \label{G0}
g_0 (x) &=& \frac{f_0 (x)}{2 \Omega_2 (x) \bar{\Omega}(x)} \left\{ \bar{\Omega}(x) \left[ \Omega_1^{\prime} (x) + \bar{\Omega}^{\prime \prime} (x)   \right] - \bar{\Omega}^{\prime} (x) \left[ \Omega_1(x) + \bar{\Omega}^{\prime} (x)   \right] \right. \nonumber \\
&& \qquad\qquad \qquad \left. + \;  \left(\frac{\bar{\Omega}(x)}{\Omega_2 (x)}\right)^2 \left[ \Omega_2^{\prime} (x)^2 - \Omega_2(x) \Omega_2^{\prime \prime} (x)  \right] \right\} \nonumber \\
&& + \; \frac{f_0^{\prime} (x)}{2 \Omega_2 (x)^2} \left\{ \Omega_2 (x) \left[ \Omega_1(x) + \bar{\Omega}^{\prime}(x)  \right] - \bar{\Omega}(x) \Omega_2^{\prime}(x) \right\} - f_0^{\prime \prime} (x) \frac{\bar{\Omega}(x)}{2 \Omega_2 (x)} \; .
\>
In summary, this procedure allows us to express both functions $g_0$ and $g_1$ in terms of $f_0$ and its derivatives. Therefore, the substitution of \eqref{G1} and \eqref{G0} in the third equation of \eqref{deter} will produce a differential equation involving solely the function $f_0$. The resulting equation reads
\[ \label{F0}
\Upsilon_0 (x) f_0 (x) + \Upsilon_1 (x) f_0^{\prime} (x) - \left[ \Omega_2 (x) \bar{\Omega} (x)  \right]^3 f_0^{\prime \prime \prime} (x) = 0
\]
with
\< 
\Upsilon_0 (x) &\coloneqq& 2 \Omega_2(x){}^4 \left[2 \Omega_0(x) \bar{\Omega}'(x) -  \bar{\Omega}(x) \Omega _0'(x)\right] 
-3 \bar{\Omega}(x){}^3 \Omega_2'(x){}^3  \nonumber \\
&& + \; \Omega_2(x) \bar{\Omega}(x){}^2 \Omega_2'(x) \left[\Omega_1(x) \Omega_2'(x) + \bar{\Omega}'(x) \Omega_2'(x) + 4 \Omega _3(x) \Omega_2''(x) \right] \nonumber \\
&& - \; \Omega_2(x){}^2 \bar{\Omega}(x) \left\{\bar{\Omega}(x) \left[\Omega_2'(x) \bar{\Omega}''(x) + \bar{\Omega}'(x) \Omega_2''(x) + \Omega_1'(x) \Omega _2'(x) + \Omega_1(x) \Omega_2''(x)\right] \right. \nonumber \\
&& \qquad \qquad \qquad \qquad \left. - \; \Omega_2'(x) \bar{\Omega}'(x) \left[\bar{\Omega}'(x) + \Omega_1(x)\right] + \Omega_2{}^{\prime \prime \prime}(x) \bar{\Omega}(x){}^2 \right\} \nonumber \\
&& - \; \Omega _2(x){}^3 \left[\bar{\Omega}'(x){}^3 - 2 \bar{\Omega}(x) \bar{\Omega}'(x) \bar{\Omega}''(x)  - \Omega_1(x){}^2 \bar{\Omega}'(x) - \bar{\Omega}(x) \Omega_1'(x) \bar{\Omega}'(x) \right. \nonumber \\
&& \quad\left. + \; \Omega_1(x) \bar{\Omega}(x) \Omega_1'(x) - 2\Omega_0(x) \bar{\Omega}(x) \Omega_2'(x) + \bar{\Omega}(x){}^2 \Omega_1''(x) + \bar{\Omega}(x){}^2 \bar{\Omega}{}^{\prime \prime \prime}(x)\right] \nonumber \\
\Upsilon_1 (x) &\coloneqq& \Omega_2(x) \Omega_3(x) \left[2\Omega_3(x) \Omega_2(x){}^2 \Omega_1'(x) - \Omega_2(x){}^2 \Omega_3'(x){}^2  - 2\Omega_1(x) \Omega_3(x) \Omega_2(x) \Omega_2'(x) \right. \nonumber \\
&& \left. \qquad\qquad\quad \; - \; 2\Omega_3(x) \Omega_2(x) \Omega_2'(x) \Omega_3'(x) + 3\Omega_3(x){}^2 \Omega_2'(x){}^2 + 2\Omega_3(x) \Omega_2(x){}^2 \Omega_3''(x) \right. \nonumber \\
&& \left. \qquad\qquad\quad \; - \; 2\Omega_3(x){}^2 \Omega _2(x) \Omega_2''(x) - 4\Omega_0(x) \Omega_2(x){}^3 + \Omega_1(x){}^2 \Omega_2(x){}^2 \right] \; . \nonumber \\
\>
Equation \eqref{F0} is a third-order linear differential equation and, therefore, it has \emph{three linearly independent solutions}. In this way, this procedure yields \emph{three vector fields} $\bm{v}_0$ spanning the symmetry algebra of \eqref{ric_gen}. It is important to remark that, although it is not clear if equation
\eqref{F0} can be explicitly integrated for generic coefficients, this has not posed as a problem for coefficients 
$\Omega_0$, $\Omega_1$, $\Omega_2$ and $\bar{\Omega}$ defined in \eqref{omegas}. We shall discuss explicit solutions in what follows.

\subsection{Generators $n=1$} \label{sec:gen1}

In order to build explicit vector fields $\bm{v}_0$ one needs to consider the explicit form of the functions 
$\bar{\Omega}(x)$, $\Omega_0(x)$, $\Omega_1(x)$ and $\Omega_2(x)$. As for the trigonometric six-vertex model 
in the $\mathfrak{h}$-module with label $n=1$ these coefficients are explicitly written down in \eqref{coeff_n1}.
Here, however, we shall consider the rational six-vertex model for simplicity reasons. The corresponding coefficients can be either obtained from \eqref{coeff_n1} in the appropriate \emph{rational limit} or be simply read off from the general formula given in \Appref{app:omega} using the statistical weights $a(x) = x+1$, $b(x)=x$ and $c(x)=1$.  

As for this particular case we do not find necessary to give details on the resolution of \eqref{F0}. 
However, as previously remarked in \Secref{sec:GEN},  Eq.  \eqref{F0} has three linearly independent solutions for generic coefficients and we can readily identify three pairs of functions $(\xi , \phi)$ solving the determining equations \eqref{V1} in this particular case. 
Each pair yields a vector field $\bm{v}_0$ and the set of vectors fields built in this manner will be named
$\mathrm{X}_{+}$, $\mathrm{X}_{-}$ and $\mathrm{H}$. They explicitly read
\< \label{XH}
\mathrm{X}_{+} &=& -\ii \frac{\partial}{\partial x} - \ii \left[ \lambda_{+}^{\prime}(x) + \left( \Lambda - \lambda_{+}(x)  \right) \frac{\lambda_{-}^{\prime}(x)}{\lambda_{-} (x)}    \right] \frac{\partial}{\partial \Lambda} \nonumber \\
\mathrm{X}_{-} &=& - \ii x^2 \frac{\partial}{\partial x} - \ii \left[ \lambda_{-}(x) + x^2 \lambda_{+}^{\prime} (x) + x \left( \Lambda - \lambda_{+}(x) \right) \left( x \frac{\lambda_{-}^{\prime} (x)}{\lambda_{-}(x)}  - 2 \right) \right] \frac{\partial}{\partial \Lambda} \nonumber \\
\mathrm{H}_{-} &=& - 2 x \frac{\partial}{\partial x} - 2 \left[ x \lambda_{+}^{\prime}(x) +  \left( \Lambda - \lambda_{+}(x) \right) \left( x \frac{\lambda_{-}^{\prime} (x)}{\lambda_{-}(x)}  - 1 \right) \right] \frac{\partial}{\partial \Lambda} \; .
\>
Now using expressions \eqref{XH} one can verify that $\mathrm{X}_{+}$, $\mathrm{X}_{-}$ and $\mathrm{H}$
span the $\mathfrak{sl}(2)$ algebra. More precisely, they satisfy the commutation relations $[\mathrm{X}_{+} , \mathrm{X}_{-}] = \mathrm{H}$ and $[\mathrm{H} , \mathrm{X}_{\pm}] = \pm 2 \mathrm{X}_{\pm}$.

We then proceed with the analysis of the symmetry transformations induced by the action of \eqref{XH}. In particular, here we write $\mathfrak{a}$ for a generator of the Lie algebra $\mathfrak{sl}(2)$ and $\mathfrak{g} = e^{\alpha \mathfrak{a}}$ 
with $\alpha \in \C^{\times}$ for the corresponding group element. The action of the symmetry generators \eqref{XH} 
then provides a map from the solution space of \eqref{ric_gen} into itself as described in \Appref{app:LIE}.
The map induced by each $\mathfrak{sl}(2)$ generator is made explicit in Table \ref{tab:MAP1}.
It is important to stress here that we have not succeeded in finding neat expressions for $\bar{\Lambda}(\bar{x})$ associated to the generators $\mathrm{X}_{-}$ and $\mathrm{H}$. Hence, we shall focus on the examination of the map obtained from the action of $\mathrm{X}_{+}$ as described in Table \ref{tab:MAP1}.

\begin{table} 
\caption{\label{tab:MAP1} $\mathfrak{sl}(2)$ maps for the case $n=1$.}
\begin{tabular}{|c|c|c|}
\hline
\rule{0pt}{2.5ex} $\mathfrak{a}$ & $\bar{x} = \mathfrak{g} \cdot x$ &  $\bar{\Lambda}(\bar{x}) = \mathfrak{g} \cdot \Lambda$ \\
\hline\hline
\rule{0pt}{2.5ex} $\mathrm{X}_{+}$ & $x - \ii \alpha$  & $\left[ \Lambda(x) - \lambda_{+}(x)  \right] \frac{\lambda_{-} (x - \ii \alpha)}{\lambda_{-} (x)} + \lambda_{+} (x - \ii \alpha)   $   \\
\hline
$\mathrm{X}_{-}$ & $x (1 + \ii \alpha x)^{-1}$   &  --  \\
\hline
$\mathrm{H}$ & $e^{-2 \alpha} x$  & -- \\
\hline
\end{tabular}
\end{table}

\begin{defi}
Let $\mathrm{Spec}_n (\mathrm{T})$ denote the set of eigenvalues of the transfer matrix $\mathrm{T}$ in the $\mathfrak{h}$-module with label $n$. Similarly, write $\mathrm{Sol}_n$ for the set of solutions of \eqref{ric_gen}.
\end{defi}

\begin{rmk}
From the construction of $\mathfrak{h}$-modules we have $\mathrm{card} \left( \mathrm{Spec}_n (\mathrm{T})  \right) = \frac{L!}{n! (L-n)!}$
and clearly $\mathrm{Spec}_n (\mathrm{T}) \subseteq \mathrm{Sol}_n$.
\end{rmk}

\begin{defi}
Considering the action of $\mathrm{X}_{+}$ described in Table \ref{tab:MAP1} we define the map 
$\mathcal{K}_1 \colon \mathrm{Sol}_1 \to \mathrm{Sol}_1$ as
\[ \label{K1}
\mathcal{K}_1 (\Lambda) \coloneqq \left[ \Lambda(x + \ii \alpha) - \lambda_{+}(x + \ii \alpha) \right] \frac{\lambda_{-}(x)}{\lambda_{-}(x + \ii \alpha)} + \lambda_{+}(x) \; .
\]
\end{defi}

\begin{rmk}
The map $\mathcal{K}_1$ depends on a continuous parameter $\alpha \in \C^{\times}$ and it has a natural action $\mathcal{K}_1 \colon \mathrm{Spec}_1 (\mathrm{T}) \to \mathrm{Sol}_1$. The specialization $\mathcal{K}_1 \colon \mathrm{Spec}_1 (\mathrm{T}) \to \mathrm{Spec}_1 (\mathrm{T})$ requires a proper fixing of the parameter $\alpha$.
\end{rmk}

\begin{prop} \label{alf}
Let $w_l, w_m \in \C$ be two zeroes of $\lambda_{-}(x)$, i.e. $\lambda_{-}( w_{l,m}) = 0$. Then the map $\mathcal{K}_1$ defined by \eqref{K1} is a morphism from $\mathrm{Spec}_1 (\mathrm{T})$ to $\mathrm{Spec}_1 (\mathrm{T})$ for
\[ \label{UAU}
\alpha = \alpha_{l m} \coloneqq \ii ( w_l - w_m) 
\]
if
\[ \label{RAU}
\Lambda(w_l) = \lambda_{+} (w_l) \qquad \mbox{and} \qquad \Lambda(w_m) \neq \lambda_{+} (w_m) \; .
\]

\end{prop}

\begin{proof}
Here we restrict our discussion to the rational model and the subset $\mathrm{Spec}_1 (\mathrm{T})$ is formed by
elements of $\mathrm{Sol}_1$ which are polynomials of maximal degree $L$. Moreover, $\lambda_{\pm} (x)$ are polynomials in $x$ of degree $L$. The image of $\mathcal{K}_1$ is then a rational function for $\Lambda(x) \in \mathrm{Spec}_1 (\mathrm{T})$ and generic parameter $\alpha$. If $\alpha$ is fixed such that the residues of \eqref{K1} at the poles 
vanish, the image of $\mathcal{K}_1$ is a polynomial. This latter requirement then leads us to conditions \eqref{UAU} and \eqref{RAU}.
\end{proof}

\begin{rmk}
An exception to Proposition \ref{alf} occurs when $\Lambda(x) = \lambda_{+}(x)$. In that case we have $\lambda_{+} \overset{\mathcal{K}_1}{\mapsto} \lambda_{+}$ for generic parameter $\alpha$.
\end{rmk}

\begin{rmk}
Conditions \eqref{RAU} can be regarded as \emph{selection rules} in the sense that they filter values of $\alpha$ 
for which $\mathcal{K}_1 \colon \mathrm{Spec}_1 (\mathrm{T}) \to \mathrm{Spec}_1 (\mathrm{T})$.
\end{rmk}

\subsubsection{Cycles in $\mathrm{Spec}_1 (\mathrm{T})$} \label{sec:cycles}
The map $\mathcal{K}_1$ defined by \eqref{K1} with parameter $\alpha$ fixed by conditions \eqref{UAU} provides a morphism from $\mathrm{Spec}_1 (\mathrm{T})$ to $\mathrm{Spec}_1 (\mathrm{T})$. Hence, since the cardinality of $\mathrm{Spec}_1 (\mathrm{T})$ is finite, $\mathcal{K}_1$ will naturally produce \emph{cycles} in the aforementioned set. In order to 
study such cycles we present in Table \ref{tab:EIGn1} eigenvalues belonging to $\mathrm{Spec}_1 (\mathrm{T})$ for lattice lengths $L = 3, 4, 5$. The latter have been obtained through the direct diagonalization of the transfer matrix $\mathrm{T}$.

\begin{table}[h]
\scalebox{0.94}{
\begin{tabular}{ | c | c | c | c | c | }
\cline{1-4}
 \multirow{2}{*}{$i$} & \multicolumn{3}{|c|}{\footnotesize{$\Lambda_i (x) \in \mathrm{Spec}_1 \left( \mathrm{T}(x) \right) $}}  \\
 \cline{2-4}
 & \scriptsize{$L=3$} & \scriptsize{$L=4$} & \scriptsize{$L=5$} \\
 \cline{1-4}
 \footnotesize{0} & \footnotesize{$\lambda_{+}(x)$} & \footnotesize{$\lambda_{+}(x)$} & \footnotesize{$\lambda_{+}(x)$} \\
 \cline{1-4}
\rule{0pt}{4.15ex} \footnotesize{1} & \footnotesize{$\frac{1}{2} \left(4 x^3 + 6 x^2 - \ii \sqrt{3} -1 \right)$}  & \footnotesize{$ 2 x^4+4 x^3+2 x^2-1$}  & 
\begin{minipage}[c]{8cm}
\footnotesize{$2 x^5+5 x^4+5x^3+\left(3+(-1)^{1/5}-(-1)^{3/5}+2 (-1)^{4/5}\right) x^2 $} \\ \footnotesize{$ +\left(1-(-1)^{3/5}+3(-1)^{4/5}\right) x+(-1)^{4/5}$} 
\end{minipage}
\\
\cline{1-4}
\rule{0pt}{4.15ex} \footnotesize{2} & \footnotesize{$\frac{1}{2} \left(4 x^3 + 6 x^2 + \ii \sqrt{3} -1 \right)$}  & \footnotesize{$2 x^4 + 4 x^3 + 2 x^2 - 2 \ii x - \ii$} & 
\begin{minipage}[c]{8cm}
\footnotesize{$2 x^5 + 5 x^4 + 5 x^3 + \left(3 - 2(-1)^{1/5} + (-1)^{2/5} - (-1)^{4/5}\right) x^2$} \\ \footnotesize{$+ \left(1 - 3(-1)^{1/5} + (-1)^{2/5} \right)x - (-1)^{1/5}$} 
\end{minipage}
\\
\cline{1-4}
\rule{0pt}{4.15ex} \footnotesize{3} & -- & \footnotesize{$2 x^4 + 4 x^3 + 2 x^2 + 2 \ii x + \ii$} &
\begin{minipage}[c]{8cm}
\footnotesize{$2 x^5 + 5 x^4 + 5 x^3 + \left(3 + 2(-1)^{2/5} + (-1)^{3/5} + (-1)^{4/5} \right) x^2 $} \\ \footnotesize{$+ \left(1 + 3(-1)^{2/5} + (-1)^{4/5}\right)x + (-1)^{2/5}$} 
\end{minipage}
\\
\cline{1-4}
\rule{0pt}{4.15ex} \footnotesize{4} & -- & -- & \begin{minipage}[c]{8cm}
\footnotesize{$2 x^5 + 5 x^4 + 5x^3 - \left(-3 + (-1)^{1/5} + (-1)^{2/5} + 2(-1)^{3/5}\right) x^2$} \\ \footnotesize{$- \left(-1 + (-1)^{1/5} + 3(-1)^{3/5}\right) x - (-1)^{3/5}$} 
\end{minipage}
\\
\cline{1-4}
\end{tabular}}
\vskip 0.25cm
\caption{\label{tab:EIGn1} \footnotesize{Elements of $\mathrm{Spec}_1 \left( \mathrm{T}(x) \right) $ for $\phi_1 = \phi_2 = 1$, $\mu_i = 0$ and lattice lengths $L = 3, 4, 5$.}}
\end{table}

\noindent \textbf{Case $L=3$.} As for this particular case we have $\mathrm{card} \left( \mathrm{Spec}_1 (\mathrm{T}) \right) = 3$ 
and the cycles generated by $\mathcal{K}_1$ are represented in Fig. \ref{fig:ORB3}. In summary, we find the transitions
$\Lambda_0 \mapsto \Lambda_0$, $\Lambda_1 \mapsto \Lambda_2$ and $\Lambda_2 \mapsto \Lambda_1$; with parameter $\alpha$ specified in Fig. \ref{fig:ORB3}.
\begin{figure} \centering
\scalebox{1}{
\begin{tikzpicture}[>=stealth]
\node at (-2, 0)  (L0) {};
\node at (0.8, 0)  (L1) {};
\node at (2.5, 0)  (L2) {};
\node at (1.7,0.8) (L12t) {\tiny{$\alpha = 1/\sqrt{3}$}};
\node at (1.7,-0.8) (L12t) {\tiny{$\alpha = -1/\sqrt{3}$}};

\fill (L0) circle [radius=2pt];
\fill (L1) circle [radius=2pt];
\fill (L2) circle [radius=2pt];

\draw (L0) node[left] {$\Lambda_0$};
\draw (L1) node[left] {$\Lambda_1$};
\draw (L2) node[right] {$\Lambda_2$};

\draw[->] (L0.north) arc (170:-155:0.5cm);
\draw[->] (L1) to [bend left=70] (L2);
\draw[->] (L2) to [bend left=70] (L1);
\end{tikzpicture}}
\caption{\footnotesize{Cycles for $L=3$.}}
\label{fig:ORB3}
\end{figure}
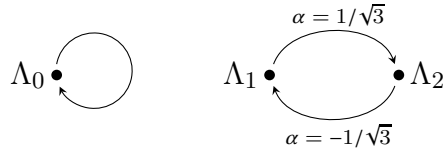

\noindent \textbf{Case $L=4$.} Proposition \ref{alf} for $L=4$ is fulfilled by $\alpha \in \{ \pm 1 , \pm \frac{1}{2} \}$. The cycles built with such values of $\alpha$ are then illustrated in Fig. \ref{fig:ORB4}.

\begin{figure} \centering
\scalebox{1}{
\begin{tikzpicture}[>=stealth]
\node at (-3, 0.5)  (L0) {};
\node at (0.5, -0.5)  (L2) {};
\node at (3, -0.5)  (L3) {};
\node at (1.75, 1.665)  (L1) {};
\fill (L0) circle [radius=2pt];
\fill (L1) circle [radius=2pt];
\fill (L2) circle [radius=2pt];
\fill (L3) circle [radius=2pt];
\draw (L0) node[left] {$\Lambda_0$};
\draw (L1) node[above] {$\Lambda_1$};
\draw (L2) node[below left] {$\Lambda_2$};
\draw (L3) node[below right] {$\Lambda_3$};
\draw[->] (L0.north) arc (170:-155:0.5cm);
\draw[->] (L1) to [bend right=20] (L2);
\draw[->] (L2) to [bend right=20] (L1);
\draw[->] (L1) to [bend right=20] (L3);
\draw[->] (L3) to [bend right=20] (L1);
\draw[->] (L2) to [bend right=20] (L3);
\draw[->] (L3) to [bend right=20] (L2);
\draw (1, 0.8) node[left] {\tiny{$\alpha_{12}$}};
\draw (1.45, 0.45) node[left] {\tiny{$\alpha_{21}$}};
\draw (3.25, 0.8) node[left] {\tiny{$\alpha_{31}$}};
\draw (2.75, 0.45) node[left] {\tiny{$\alpha_{13}$}};
\draw (1.75, -0.3) node[above] {\tiny{$\alpha_{32}$}};
\draw (1.75, -0.7) node[below] {\tiny{$\alpha_{23}$}};
\end{tikzpicture}}
\caption{\footnotesize{Cycles for $L=4$. The parameters $\alpha$ are given by $\alpha_{12} = -1/2 = - \alpha_{21}$,
$\alpha_{13} = 1/2 = - \alpha_{31}$ and $\alpha_{23} = 1 = - \alpha_{32}$.}}
\label{fig:ORB4}
\end{figure}
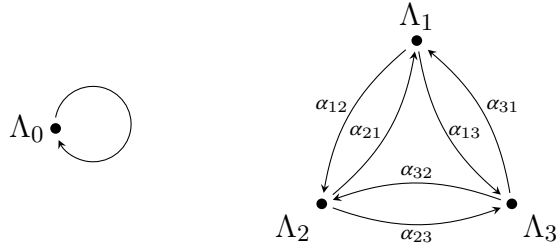

\noindent \textbf{Case $L=5$.} The elements of $\mathrm{Spec}_1 (\mathrm{T})$ for $L=5$ acquires a more involving structure as can be seen in Table \ref{tab:EIGn1}. Nevertheless, the flow generated by the action of $\mathcal{K}_1$ can still be studied with analytical expressions for the parameter $\alpha$. According to Proposition \ref{alf} we then have
\begin{align} \label{alphas}
\alpha_{12} &= - \alpha_{21} = -\sqrt{1-\frac{2}{\sqrt{5}}} & \alpha_{23} &= - \alpha_{32} = \frac{1}{2} \left( \sqrt{1-\frac{2}{\sqrt{5}}} + \sqrt{1+\frac{2}{\sqrt{5}}} \right) \nonumber \\
\alpha_{13} &= - \alpha_{31} = -\frac{1}{2} \left( \sqrt{1-\frac{2}{\sqrt{5}}} -\sqrt{1+\frac{2}{\sqrt{5}}} \right) \quad & \alpha_{24} &= - \alpha_{42} = \frac{1}{2} \left( \sqrt{1-\frac{2}{\sqrt{5}}} - \sqrt{1+\frac{2}{\sqrt{5}}} \right) \nonumber \\
\alpha_{14} &= - \alpha_{41} = - \frac{1}{2} \left( \sqrt{1-\frac{2}{\sqrt{5}}} + \sqrt{1+\frac{2}{\sqrt{5}}} \right) \quad & \alpha_{34} &= - \alpha_{43} = -\sqrt{1+\frac{2}{\sqrt{5}}} \; , \nonumber \\
\end{align}
and the corresponding cycles are illustrated in Fig. \ref{fig:ORB5}.

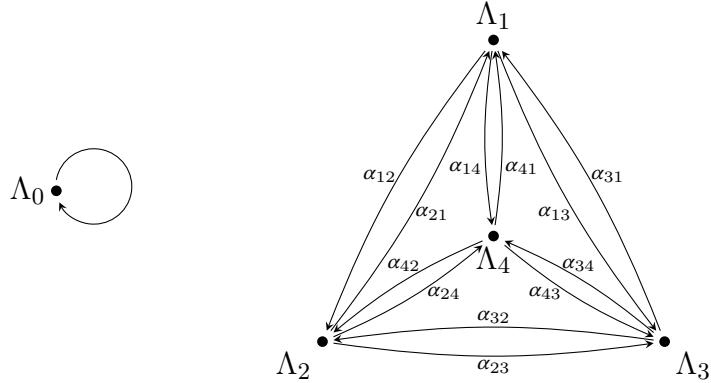
\begin{figure} \centering
\scalebox{1}{
\begin{tikzpicture}[>=stealth]
\node at (-3, 1.5)  (L0) {};

\node at (2.75, 3.5)  (L1) {};

\node at (0.5, -0.5)  (L2) {};
\node at (5, -0.5)  (L3) {};
\node at (2.75, 0.9)  (L4) {};

\fill (L0) circle [radius=2pt];
\fill (L1) circle [radius=2pt];
\fill (L2) circle [radius=2pt];
\fill (L3) circle [radius=2pt];
\fill (L4) circle [radius=2pt];

\draw (L0) node[left] {$\Lambda_0$};
\draw (L1) node[above] {$\Lambda_1$};
\draw (L2) node[below left] {$\Lambda_2$};
\draw (L3) node[below right] {$\Lambda_3$};
\draw (L4) node[below] {$\Lambda_4$};

\draw[->] (L0.north) arc (170:-155:0.5cm);

\draw[->] (L1) to [bend right=8] (L2);
\draw[->] (L2) to [bend right=8] (L1);
\draw[->] (L1) to [bend right=8] (L3);
\draw[->] (L3) to [bend right=8] (L1);
\draw[->] (L2) to [bend right=8] (L3);
\draw[->] (L3) to [bend right=8] (L2);

\draw[->] (L1) to [bend right=8] (L4);
\draw[->] (L4) to [bend right=8] (L1);
\draw[->] (L2) to [bend right=8] (L4);
\draw[->] (L4) to [bend right=8] (L2);
\draw[->] (L3) to [bend right=8] (L4);
\draw[->] (L4) to [bend right=8] (L3);

\draw (2.75, 1.8) node[left] {\tiny{$\alpha_{14}$}};
\draw (2.75, 1.8) node[right] {\tiny{$\alpha_{41}$}};

\draw (1.2, 0.52) node[right] {\tiny{$\alpha_{42}$}};
\draw (1.73, 0.15) node[right] {\tiny{$\alpha_{24}$}};

\draw (3.5, 0.52) node[right] {\tiny{$\alpha_{34}$}};
\draw (3.06, 0.15) node[right] {\tiny{$\alpha_{43}$}};

\draw (2.75, -0.4) node[above] {\tiny{$\alpha_{32}$}};
\draw (2.75, -0.6) node[below] {\tiny{$\alpha_{23}$}};

\draw (1.57, 1.2) node[right] {\tiny{$\alpha_{21}$}};
\draw (3.185, 1.2) node[right] {\tiny{$\alpha_{13}$}};

\draw (1.62, 1.7) node[left] {\tiny{$\alpha_{12}$}};
\draw (4.62, 1.7) node[left] {\tiny{$\alpha_{31}$}};

\end{tikzpicture}}
\caption{\footnotesize{Cycles for $L=5$. The parameters $\alpha_{ij}$ are explicitly given in \eqref{alphas}.}}
\label{fig:ORB5}
\end{figure}

\subsection{Generators $n=2$} \label{sec:gen2}

Along the same lines discussed in \Secref{sec:gen1}, we now address the question of constructing vector fields $\bm{v}_0$ describing
continuous symmetries in $\mathrm{Sol}_2$. In fact, as far as this particular case is concerned we only need to specialize the general formalism presented in \Secref{sec:GEN}. Although there are no fundamental obstacles in that case, there is actually a practical caveat: 
the expressions obtained for the pair of functions $\left( \xi, \phi \right)$ in the case $n=2$ are too long to be fully reproduced in this subsection. However, the existence of three linearly independent vector fields, which shall also be named $\mathrm{X}_{+}$, $\mathrm{X}_{-}$ and $\mathrm{H}$, is granted since this property holds true for arbitrary values of the label $n$ in the interval $1 \leq n \leq L$. The latter property has been explained in \Secref{sec:GEN} and follows from the fact that \eqref{F0} is a third-order linear differential equation. Moreover, in the case $n=2$ the coefficients $\bar{\Omega}(x)$, $\Omega_0 (x)$, $\Omega_1 (x)$ and $\Omega_2 (x)$, and consequently the vector fields $\mathrm{X}_{\pm}$ and $\mathrm{H}$, depend additionally on a parameter $u_1 \in \C$ which is one of the zeroes of the solution $\Lambda$.

As vector fields the elements $\mathrm{X}_{\pm}$ and $\mathrm{H}$ built out of solutions of the conditions \eqref{V1} are expected to close a Lie algebra. In that case we would be dealing with a Lie algebra spanned by three generators and $\mathfrak{sl}(2)$ would be the most natural candidate. The latter argument justifies naming them $\mathrm{X}_{\pm}$ and $\mathrm{H}$; and one can actually verify the vector fields $\bm{v}_0$ built in this case indeed satisfy the $\mathfrak{sl}(2)$ commutation relations.
Such symmetry generators read
\<
\mathrm{H} &=& \xi_{\mathrm{H}}(x, \Lambda) \frac{\partial}{\partial x} + \phi_{\mathrm{H}}(x, \Lambda) \frac{\partial}{\partial \Lambda} \nonumber \\
\mathrm{X}_{\pm} &=& \xi_{\pm}(x, \Lambda) \frac{\partial}{\partial x} + \phi_{\pm}(x, \Lambda) \frac{\partial}{\partial \Lambda} 
\>
with
\<
\xi_{\mathrm{H}}(x, \Lambda) &\coloneqq& \frac{2 x \left[(x-u_1) \lambda_{+}(u_1) - \lambda_{-}(u_1) \right] \left[(x - 2 u_1) \lambda_{-}(u_1) - (u_1^2- u_1 x + 1 ) \lambda_{+}(u_1)\right]}{\left[\lambda_{-}(u_1) + u_1 \lambda_{+}(u_1) \right] \left[2 (u_1-x) \lambda_{-}(u_1) +  (u_1^2 - 2 u_1 x + x^2 + 1 ) \lambda_{+}(u_1) \right]} \nonumber \\ 
\xi_{+}(x, \Lambda) &\coloneqq& \frac{ x^2 \left[ (x - 2 u_1) \lambda_{-}(u_1) - (u_1^2 - u_1 x + 1 ) \lambda_{+}(u_1) \right]^2}{\left[2 u_1 \lambda_{-}(u_1) + (u_1^2 + 1 ) \lambda_{+}(u_1) \right]^2 \left[2 (u_1 - x) \lambda_{-}(u_1) + (u_1^2 - 2 u_1 x + x^2 + 1 ) \lambda_{+}(u_1)  \right]} \nonumber \\
\xi_{-}(x, \Lambda) &\coloneqq& -\frac{ \left[2 u_1 \lambda_{-}(u_1) + (u_1^2 + 1) \lambda_{+}(u_1) \right]^2  \left[\lambda_{-}(u_1) + (u_1-x) \lambda_{+}(u_1) \right]^2}{ \left[ \lambda_{-}(u_1) + u_1 \lambda_{+}(u_1) \right]^2 \left[2 (u_1-x) \lambda_{-}(u_1) + (u_1^2 - 2 u_1 x + x^2 + 1) \lambda_{+}(u_1) \right]} \; . \nonumber \\
\>
The expressions for the functions $\phi_{\mathrm{H}}$ and $\phi_{\pm}$ are in their turn rather lengthy to be presented here and we leave them to \Appref{app:PHI}
merely for the sake of completeness. It is important to remark that these functions represent the main setback for obtaining neat expressions similar to the ones presented in Table \ref{tab:MAP1} for the case $n=1$. Nevertheless, the formal maps $\bar{x} = e^{\alpha \mathfrak{a}} \cdot x$ and 
$\bar{\Lambda}(\bar{x}) = e^{\alpha \mathfrak{a}} \cdot \Lambda$ for $\mathfrak{a} \in \{ \mathrm{X}_{\pm} , \mathrm{H} \}$ still hold.

\section{Concluding remarks} \label{sec:CONCL}

In this paper we have continued the study of the relations between non-linear differential equations and the six-vertex model previously
uncovered in \cite{Galleas_2017}. The results of the present paper focus on two different but complementary analysis of the six-vertex model
eigenvalue problem. Both analysis originate from the Riccati representation of the six-vertex model's spectral problem which was firstly put forward in \cite{Galleas_2017} and extended in \Secref{sec:RIC} of the present work. The necessity of such extension is due to the fact the Riccati equation
presented in \cite{Galleas_2017} was meant to describe only a fraction of the eigenvalues of the six-vertex model's transfer matrix. In order to 
encode the full spectrum into Riccati equations some technical difficulties had to be first overcome. This problem is then addressed in \Secref{sec:RIC} of the present paper where we have also obtained a Riccati representation covering the full spectrum.

One of the aforementioned branches of analysis that we carry on in this work is already included and discussed in \Secref{sec:RIC}. That consists in describing the model's spectrum by means of algebraic equations whose roots can be converted into explicit expressions for the sought eigenvalues. This approach is analogous to the standard Bethe ansatz \emph{solution} of vertex models exhibiting commuting transfer matrices; but the roots we use here have a completely different nature. While in the Bethe ansatz framework such auxiliary variables (Bethe roots) can be traced back to the zeroes of eigenvalues of the so called \textsc{q}-operators, the aforementioned Riccati equations allows one to write down algebraic equations characterizing the zeroes of the transfer matrix's eigenvalues themselves. In this way, we end up with a very compact representation for such eigenvalues as can be seen from \eqref{LL}. The system of algebraic equations characterizing the model's spectrum is then described in Lemma \ref{ZEROESU}.   

The second type of analysis that we offer in this work concerns the existence of continuous symmetries in the spectrum of the six-vertex model. This analysis is performed in \Secref{sec:GEN} using the Lie groups method explained in \Appref{app:LIE}. The Lie algebra 
$\mathfrak{sl}(2)$ is known to play an important role in the six-vertex model and, as far as vertex models are concerned, we find 
an unusual action of such algebra in \Secref{sec:GEN}. For instance, it appears as the algebra realizing continuous symmetries in the transfer matrix's spectrum. In this particular setup continuous symmetries are then regarded as a map from the space of the transfer matrix's eigenvalues into itself. In this way, one can use such maps to relate all eigenvalues in a given $\mathfrak{h}$-module by means of an operator. This analysis has been explicitly carried out for the simplest non-trivial $\mathfrak{h}$-module in \Secref{sec:gen1};
where we also find that such realization can be conveniently pictured as a (colored) directed graph with vertices representing the 
eigenvalues and (colored) edges directed from one vertex to another representing the action of the group symmetry generator. 
From that we can also see the appearance of \emph{cycles} in the spectrum which are then depicted in Figures \ref{fig:ORB3}, \ref{fig:ORB4} and \ref{fig:ORB5}
for small values of the lattice length.

The explicit derivation of the symmetry generators in the $\mathfrak{h}$-module $n=2$ is described in \Secref{sec:gen2} but in that case 
we do not find compact expressions for them. Due to that we were unable to write down explicit expressions for the symmetry maps and proceed with the inspection of cycles along the same lines of \Secref{sec:cycles}. We hope we are able to report on that problem in a future publication.

\bibliographystyle{alpha}
\bibliography{references}

\newpage
\appendix

\section{Symmetry groups of differential equations} \label{app:LIE}

The origin of Lie groups is intimately associated with the study of \emph{Ordinary Differential Equations} (ODEs). It was introduced by Sophus Lie 
in the second half of the $19^{\mathrm{th}}$ century aiming to understand the several existing methods for solving ODEs from an unified perspective.
The work of Lie was largely inspired by Galois' theory of algebraic equations and it can be regarded as an analogue of Galois's theory for 
differential equations. In this Appendix we intend to give a brief account of the Lie groups method for differential equations and introduce some concepts/notation employed in the main text.

The method of Lie symmetries of differential operators has a large geometrical content and it regards a differential equation as an embedding of a family
of curves living in a space of independent and dependent variables. Then, within Lie's framework, each curve forming this family corresponds to a solution
of the associated differential equation. Here we shall restrict our discussion to ODEs, although Lie's method can also be formulated for 
\emph{Partial Differential Equations} (PDEs).  In this way, suppose we are considering an ODE describing a function $f(x)$. We then refer to $x$ as
independent variable while $f(x)$ and its derivatives with respect to $x$ are regarded as dependent variables.

The general idea underlying Lie groups method for differential equations is rather simple and it can be easily formulated 
in terms of geometrical objects. For that we firstly recall that any ODE of order $n$ can be recasted as the surface equation
\[ \label{sur}
\Sigma (x, f(x), f^{(1)} (x) , \dots , f^{(n)} (x) ) = 0 
\]
for a given \emph{differential function} $\Sigma$. Hence, the Lie groups method consists in finding infinitesimal group transformations leaving the 
surface $\Sigma$ invariant. Such transformations will then map the manifold of solutions of a given ODE into itself. In order to proceed let us
introduce the following notions and concepts.

\subsection{Vector fields} \label{vfields}
Let $z$ be a point in the smooth manifold $\mathcal{M}$ and let $\Omega \subseteq \mathcal{M}$ be a (smooth) curve in $\mathcal{M}$ such 
that $\Omega \ni z$. We then write $\bm{v}_{\Omega} (z)$ for the vector tangent to the curve $\Omega$ at the point $z$.
Next let $\{ \Omega_i \}$ be the set of all curves $\Omega_i \subseteq \mathcal{M}$ containing the point  $z \in \mathcal{M}$ and write 
$ T \mathcal{M} (z) \coloneqq \{ \bm{v}_{\Omega_i} (z) \colon \Omega_i \ni z  \}$ for the \emph{tangent space} to $\mathcal{M}$ at $z$. 
The space $T \mathcal{M} (z)$ is a vector space while $T \mathcal{M} \coloneqq \cup_{z \in \mathcal{M}}  T \mathcal{M} (z)$ corresponds to a
\emph{tangent bundle} built in standard manner. 

Now suppose $\mathcal{M}$ is $m$-dimensional and write $z = (z_1, z_2 , \dots , z_m)$ in terms of local coordinates $z_i$. In this way a
\emph{vector field} $\bm{v}$, regarded as a varying assignment of tangent vectors, can be written as 
\[ \label{vf}
\bm{v} \coloneqq \sum_{i=1}^m \xi_i (z) \frac{\partial}{\partial z_i} 
\]
for given analytic functions $\xi_i$. Moreover, we shall employ the notation $\bm{v} (a)$ to denote the action of the vector field 
$\bm{v}$ on an element $a \in \mathcal{M}$. Hence, in local coordinates one finds
\[ \label{exp}
\text{exp}(\alpha \; \bm{v}) (z) = z + \alpha  \bm{v} (z) + \text{O}(\alpha^2) 
\]
for a generic parameter $\alpha$. 

\begin{rema}
In \emph{fluid mechanics} the vector field $\bm{v}$ is usually refereed to as \emph{fluid velocity} while $\text{exp}(\alpha \; \bm{v})$
corresponds to the \emph{flow} generated by $\bm{v}$.
\end{rema}

Interestingly, one can also construct the tangent space $T \mathcal{M} (z)$ from the associated flows as a consequence of \eqref{exp}. More precisely we have
\[ \label{VT}
\bm{v}_{\Omega} (z) =  \left. \frac{\dd}{\dd \alpha} \text{exp}(\alpha \; \bm{v}) (z) \right|_{\alpha=0} = \bm{v}(z)   \qquad z \in \mathcal{M}   \; ,
\]
which allows one to identify tangent vectors in $T \mathcal{M} (z)$ with ordinary vectors in $\mathcal{M}$.

\begin{defi}[Lie brackets]
Let $\bm{v}$ and $\bm{w}$ be vectors fields acting on $\mathcal{M}$. Their \emph{Lie brackets} or \emph{commutator} $[\bm{v} , \bm{w}]$ is then defined
through the action $[\bm{v} , \bm{w}] (f) \coloneqq v(w (f)) - w(v(f))$ for any smooth function $f$ on $\mathcal{M}$.
\end{defi}

Lie brackets enjoy several important properties. For instance, from \eqref{vf} one can readily show that commutators are also vector fields. Moreover, they
are anti-symmetric and bilinear by definition; and additionally satisfy the \emph{Jacobi identify} 
\[
\left[ \bm{u} , [\bm{v} , \bm{w}]  \right] + \left[ \bm{v} , [\bm{w} , \bm{u}]  \right] + \left[ \bm{w} , [\bm{u} , \bm{v}]  \right] = 0
\]
for any triplet of vector fields $\bm{u}$, $\bm{v}$ and $\bm{w}$.

\subsection{Jet spaces, group action and prolongations} 

In order to precise the meaning of \emph{symmetry groups of differential equations} we first need to declare the space on which 
the group under consideration acts. As far as differential equations are concerned, such groups act on the space spanned by both independent and dependent variables relevant to the differential equation. In the present work we restrict our analysis to differential equations governing functions
$\Lambda \colon \C \to \C$. More precisely, $\Lambda$ depends on a single complex variable $x$ and we write  $\Lambda^{(n)} (x) \coloneqq \frac{\dd^n \Lambda(x) }{\dd x^n}$ for the $n$-th derivative of $\Lambda$ with respect to $x$. Hence, the total space enclosing both independent and dependent 
variables is $\mathbb{J} \coloneqq \C \times \C$. The latter is refereed to as \emph{jet space} and it is the space where symmetry groups will act.

Now let $\mathfrak{G}$ be a group and let $\mathfrak{g} \in \mathfrak{G}$ denote its generators. We also let $\mathfrak{G}$ act as diffeomorphisms
on $\J$ as we are interested on continuous symmetries in the present work. In this way we consider the group transformation
\[ \label{GT}
(\bar{x} , \bar{\Lambda}) \coloneqq \mathfrak{g} \cdot (x , \Lambda) \eqqcolon (\chi (x, \Lambda) , \psi (x, \Lambda)) 
\]
for given continuous functions $\chi$ and $\psi$. The transformation \eqref{GT} acts pointwise on $\J$ and it is usually refereed to as
\emph{point transformations}.

The jet space $\J$ is not large enough to accommodate a differential equation and this issue can be understood as follows. Since any differential equation can be regarded as a surface equation $\Sigma = 0$, even in the simplest case scenario one has $\Sigma$ depending on $x$ (independent variable), $\Lambda(x)$ and $\Lambda^{(1)} (x)$ (dependent variables). As for higher-order equations 
$\Sigma$ will then depend on more variables. Hence, we can readily see that $\J$ is not able to accommodate $\Sigma$ and \emph{prolongations} will then provide an appropriate space to fit any differential function of interest.

We shall start by defining prolongations of the function $\Lambda$. In that case we say a smooth function $\Lambda \colon \C \to \C$
has a $n$-th prolongation given by
\[
\mathrm{Pr}^{(n)} \Lambda \coloneqq ( \Lambda(x) , \Lambda^{(1)} (x) , \dots , \Lambda^{(n)} (x)) \; .
\]
Hence, $\mathrm{Pr}^{(n)} \Lambda$ is a function on $\C$ with values in $\C^{n+1}$. We then similarly write 
$\mathrm{Pr}^{(n)} \J \coloneqq \C \times \C^{n+1}$ for the $n$-th jet space whose points are the $(n+2)$-tuples
$(x, \mathrm{Pr}^{(n)} \Lambda)$. For latter convenience we also write $\Gamma_{\Lambda} \coloneqq \{ ( x , \Lambda(x) ) \in \J \}$ for the graph of the function $\Lambda$ and $\mathrm{Pr}^{(n)} \Gamma_{\Lambda} \coloneqq \{ ( x , \mathrm{Pr}^{(n)} \Lambda ) \in \mathrm{Pr}^{(n)} \J \}$ for its $n$-th prolongation.

Next we turn our attention to the transformation group $\mathfrak{G}$ taking into account the concept of prolongation. In this sense,
if $\mathfrak{g} \in \mathfrak{G}$ is a point transformation in accordance with \eqref{GT}, then $\mathfrak{g}$ acts on a 
function $\Lambda$ by transforming its graph $\Gamma_{\Lambda}$ and consequently its prolongation $\mathrm{Pr}^{(n)} \Gamma_{\Lambda}$.
This reasoning motivates the definition of induced prolonged transformations on the $n$-th jet space $\mathrm{Pr}^{(n)} \J$. More precisely, we then write
\[
( \bar{x} , \mathrm{Pr}^{(n)} \bar{\Lambda} ) \eqqcolon  \mathrm{Pr}^{(n)} \mathfrak{g} \cdot ( x , \mathrm{Pr}^{(n)} \Lambda ) 
\]
where a point $z \in \mathrm{Pr}^{(n)} \J$ is transformed into the point $\bar{z} = \mathrm{Pr}^{(n)} \mathfrak{g} \cdot z$
by evaluating the derivatives of the transformed function $\bar{\Lambda} = \mathfrak{g} \cdot \Lambda$ with respect to the variable
$\bar{x} = \mathfrak{g} \cdot x$.

\subsection{Differential functions and derivatives}
In order to study group symmetries of a differential equation one needs to understand how group transformations act on derivatives of dependent variables, in addition to knowing how the group acts on independent and dependent variables. That is the point where the 
$n$-th jet space $\mathrm{Pr}^{(n)} \J$ comes in handy by allowing all the required derivatives to be worked out on the same footing.
As far as our terminology is concerned, it is also worth giving a more formal definition of \emph{differential functions}.

\begin{defi}
We call differential function of order $n$ a smooth function $f \colon \mathrm{Pr}^{(n)} \J \to \C$ defined on an open
subset of the $n$-th jet space $\mathrm{Pr}^{(n)} \J$.
\end{defi}

\begin{rema}
Any differential equation of order $n$ can be recasted as the vanishing condition of a $n$-th order differential function.
\end{rema}

Next we need to introduce derivatives in $\mathrm{Pr}^{(n)} \J$. For that we let $f(x, \mathrm{Pr}^{(n)} \J )$ be a differential function of order $n$ and  write $\DD_x f$ for the total derivative of $f$ with respect to $x$.
More precisely, we define it as
\[ \label{DX}
\DD_x f \coloneqq \frac{\partial f}{\partial x} + \sum_{i=0}^n \Lambda^{(i+1)} \frac{\partial f}{\partial \Lambda^{(i)}} 
\]
in such a way that, if $f$ is a differential function of order $n$, then $\DD_x f$ is a differential function of order $n+1$.

One can readily recognize \eqref{DX} as the usual total derivative. However, it is convenient to define the compact notation
$\DD_x$ since we will also need to evaluate the prolonged elements $\mathrm{Pr}^{(n)} \bar{\Lambda}$. That is, one needs to
compute $\bar{\Lambda}^{(1)}$, $\bar{\Lambda}^{(2)}$ and so on; and such elements are defined as
\<
\bar{\Lambda}^{(n)} \coloneqq \frac{\dd^n \bar{\Lambda}}{\dd \bar{x}^n} \; .
\>
Hence, taking into account the point transformation \eqref{GT}, we have
\[
\frac{\dd \bar{\Lambda}}{\dd \bar{x}} = \frac{\DD_x \psi}{\DD_x \chi} 
\]
and 
\[
\frac{\dd^2 \bar{\Lambda}}{\dd \bar{x}^2} = \frac{1}{\DD_x \chi} \DD_x \left( \frac{\DD_x \psi}{\DD_x \chi} \right) = \frac{(\DD_x \chi) (\DD_x^2 \psi) - (\DD_x \psi) (\DD_x^2 \chi)}{(\DD_x \chi)^3} \; .
\]
In this way, for practical purposes we only need to replace
\[
\frac{\dd}{\dd \bar{x}} \mapsto \ \frac{1}{\mathrm{D}_x \chi} \mathrm{D}_x \qquad \mbox{and} \qquad \bar{\Lambda} \mapsto \psi \; .
\]

\subsection{Prolongation of vector fields}
Let $\mathfrak{G}$ be a group acting on $\J$ as an one-parameter point transformation through generators 
$\mathfrak{g} \coloneqq \text{exp}(\alpha \; \bm{v}_0) \in \mathfrak{G}$, with  $\bm{v}_0$ a vector field on $\J$. Hence we write
\[ \label{v0}
\bm{v}_0 \coloneqq \xi(x, \Lambda) \frac{\partial}{\partial x} + \phi (x, \Lambda) \frac{\partial}{\partial \Lambda} 
\]
for given analytic functions $\xi(x, \Lambda)$ and $\phi(x, \Lambda)$. However, as here we are interested in investigating point transformations
of the surface $\Sigma = 0$, which lives in $\mathrm{Pr}^{(n)} \J$, we shall need to build suitable prolongations of the vector field $\bm{v}_0$.
In this way we define the prolonged vector field $\bm{v}_n$ as the infinitesimal generator of the prolonged one-parameter group element
$\mathrm{Pr}^{(n)} \mathfrak{g}$. More precisely, we consider an arbitrary point $z = (x , \mathrm{Pr}^{(n)} \Lambda) \in \mathrm{Pr}^{(n)} \J $ 
and write $\mathrm{Pr}^{(n)} \mathfrak{g} \eqqcolon \text{exp} (\alpha \; \bm{v}_n)$ in such a way that
\[ \label{vnz}
\bm{v}_n (z) = \left. \frac{\dd}{\dd \alpha} \mathrm{Pr}^{(n)} \mathfrak{g} \cdot z \right|_{\alpha=0} \; .
\]
An explicit formula for $\bm{v}_n$ then follows from \eqref{vnz} and it reads
\[ \label{vn}
\bm{v}_n = \bm{v}_0 + \sum_{m=1}^n \phi_m \frac{\partial}{\partial \Lambda^{(m)}}
\]
with $\phi_m \coloneqq \DD_x^m \left( \phi - \xi \Lambda^{(1)} \right) + \xi \; \Lambda^{(m+1)}$.

\begin{eg}
Let us introduce the symbols $\xi_{z_1 z_2 \dots z_l} \coloneqq \frac{\partial}{\partial z_l} \dots \frac{\partial}{\partial z_2} \frac{\partial}{\partial z_1} \xi(x, \Lambda)$ and $\phi_{z_1 z_2 \dots z_l} \coloneqq \frac{\partial}{\partial z_l} \dots \frac{\partial}{\partial z_2} \frac{\partial}{\partial z_1} \phi(x, \Lambda)$ for $z_j \in \{ x , \Lambda \}$. In this way we find for instance
\<
\phi_1 &=& \phi_x + (\phi_{\Lambda} - \xi_{x}) \Lambda^{(1)} - \xi_{\Lambda} \left( \Lambda^{(1)} \right)^2 \nonumber \\
\phi_2 &=& \phi_{x x} + (2 \phi_{x \Lambda} - \xi_{xx}) \Lambda^{(1)} + (\phi_{\Lambda \Lambda} - 2 \xi_{x \Lambda}) \left( \Lambda^{(1)} \right)^2 \nonumber \\
&& - \; \xi_{\Lambda \Lambda} \left( \Lambda^{(1)} \right)^3 + (\phi_{\Lambda} - 2 \xi_{x}) \Lambda^{(2)} - 3 \xi_{\Lambda} \Lambda^{(1)} \Lambda^{(2)} \; .
\>
\end{eg}

\begin{rema} 
As discussed in \Secref{vfields} the commutator $\left[ \bm{v}_0 , \bm{w}_0 \right]$ is also a vector field. Consequently, it 
admits a $n$-th order prolongation which we refer to as $\left[ \bm{v}_0 , \bm{w}_0 \right]_n$. However, as the prolongation process respects
the composition of maps, one can readily show that
\[
\left[ \bm{v}_0 , \bm{w}_0 \right]_n = \left[ \bm{v}_n , \bm{w}_n \right] \; .
\]
In this way, the prolongation process is an algebra homomorphism from the space of vector fields on $\J$ to the space of vector 
fields on $\mathrm{Pr}^{(n)} \J$. Therefore, the $n$-th prolongation $\mathrm{Pr}^{(n)} \mathfrak{g}$ also defines a Lie algebra which
generates the $n$-th prolongation of the associated group of transformations.
\end{rema}

\subsection{Symmetries} 
At this stage it is important to clarify the notion of symmetry that we have considered throughout this paper. 

\begin{defi}[Symmetry]
Let the vanishing of the differential function $\Sigma$, i.e. 
\[ \label{SUR}
\Sigma (x , \Lambda , \Lambda^{(1)} , \dots , \Lambda^{(n)}) = 0 \; ,
\]
characterize the differential equation of interest. The point transformation $\mathfrak{g} \colon \J \to \J$ is then called a \emph{symmetry} 
of $\Sigma$ if the transformed function $\bar{\Lambda} = \mathfrak{g} \cdot \Lambda = \bar{\Lambda} (\bar{x})$ with $\bar{x} = \mathfrak{g} \cdot x$
solves \eqref{SUR}, whenever $\Lambda(x)$ is a solution. More precisely, $\mathfrak{g}$ is a symmetry if 
\[
\Sigma(\bar{x}, \bar{\Lambda}, \bar{\Lambda}^{(1)} , \dots , \bar{\Lambda}^{(n)} ) = 0 \quad \mathrm{whenever} \quad \Sigma(x, \Lambda, \Lambda^{(1)} , \dots , \Lambda^{(n)} ) = 0 \; .
\]
\end{defi}

\begin{rema}
The notion of symmetry employed here is defined in terms of group actions on the space of solutions of \eqref{SUR}.
\end{rema}

\begin{thm}[Lie] \label{LIE}
Let $\mathfrak{G}$ be a connected group of point transformations acting on $\J$ with elements $\mathfrak{g} = \mathrm{exp}(\alpha \; \bm{v}_0) \in \mathfrak{G}$ and $\bm{v}_0$ a vector field according to the prescription \eqref{v0}. Also, let $\bm{v}_n$ be the $n$-th prolongation of $\bm{v}_0$ 
as given by formula \eqref{vn}. Then $\mathfrak{G}$ is a symmetry group of the differential equation $\Sigma = 0$ iff 
\[ \label{lie}
\bm{v}_n \left( \Sigma  \right) = 0 \qquad \text{whenever} \qquad \Sigma = 0 \; ,
\]
for all vector fields $\bm{v}_0$ spanning the algebra corresponding to the group $\mathfrak{G}$.
\end{thm}

As far as the implementation of Theorem \ref{LIE} is concerned some comments are in order. For instance, although Eq. \eqref{lie} consists of two conditions, the second condition is simply telling us that the first one only needs to hold on points $z \in \mathrm{Pr}^{(n)} \J$ on the surface $\Sigma$.
In this way, the implementation of Theorem \ref{LIE} firstly requires us to substitute formula \eqref{vn} in the first relation of \eqref{lie}. Then we use the condition $\Sigma = 0$ to find an expression for $\Lambda^{(n)}$ in terms of variables spanning $\mathrm{Pr}^{(n-1)} \J$. Such expression for 
$\Lambda^{(n)}$ is then substituted back in the result obtained from the action of \eqref{vn} on $\Sigma$ according to the first relation of
\eqref{lie}. By doing so we are left with a finite series expansion in variables $\prod_{i=1}^{n-1} \left( \Lambda^{(i)}  \right)^{m_i}$ with powers $m_i \in \Z$. The coefficients of such expansion will then depend on $x$, $\Lambda$ and partial derivatives of the functions $\xi$ and $\phi$. 
Hence, in order to having $\bm{v}_n \left( \Sigma  \right) = 0$ fulfilled we need such coefficients to vanish. The latter forms a system of 
\emph{over-determined} linear PDEs for the functions $\xi(x, \Lambda)$ and $\phi(x, \Lambda)$ which are usually
refereed to as \emph{determining equations}. Although this is not guaranteed, determining equations are usually simple enough such that their exact solution can be found using elementary methods. In this way, one can completely determine the (connected) symmetry group of a differential equation by following the above described algorithm. 

\newpage

\section{Matrix entries $\omega_{i,j}$} \label{app:omega}

It is convenient to first introduce some extra conventions in order to define the matrix entries $\omega_{i,j}$.
Write $\mathscr{Q} \coloneqq \{ 0, 1 \}$ and for $i \in \mathscr{Q}$ define its complement $\bar{\imath} = 1 - i$. Hence, $\bar{\imath} \in \mathscr{Q}$
if $i \in \mathscr{Q}$. The aforementioned coefficients are then given by
\<
&&\omega_{i,j} \coloneqq \nonumber \\
&&\begin{cases}
\displaystyle (-1)^i \left[ \phi_1 a(x_{\bar{\imath}} - x_i) \prod_{k=1}^{n-1} \frac{a(u_k - x_i)}{b(u_k - x_i)} \lambda_{\mathcal{A}}(x_i) - \phi_2 a(x_i - x_{\bar{\imath}}) \prod_{k=1}^{n-1} \frac{a(x_i - u_k)}{b(x_i - u_k)} \lambda_{\mathcal{D}}(x_i) \right] \cr
\hfill i \in \mathscr{Q} \; , \; j=i \cr
\displaystyle (-1)^i c \left[ \phi_1 \prod_{k=1}^{n-1} \frac{a(u_k - x_i)}{b(u_k - x_i)} \lambda_{\mathcal{A}}(x_i) - \phi_2 \prod_{k=1}^{n-1} \frac{a(x_i - u_k)}{b(x_i - u_k)} \lambda_{\mathcal{D}}(x_i) \right] \;\;\; \qquad\qquad\quad i \in \mathscr{Q} \; , \; j=\bar{\imath} \cr
\displaystyle (-1)^{i+1} \left[ \phi_1 a(x_{\bar{\imath}} - x_i) \frac{c(x_i - u_{j-1})}{b(x_i - u_{j-1})} \prod_{\substack{k=1 \\ k \neq j-1}}^{n-1} \frac{a(u_k - x_i)}{b(u_k - x_i)} \lambda_{\mathcal{A}}(x_i) \right. \cr
\displaystyle \qquad \qquad \qquad - \; \left. \phi_2 a(x_i - x_{\bar{\imath}}) \frac{c(u_{j-1} - x_i)}{b(u_{j-1} - x_i)} \prod_{\substack{k=1 \\ k \neq j-1}}^{n-1} \frac{a(x_i - u_k)}{b(x_i - u_k)} \lambda_{\mathcal{D}}(x_i) \right] \hfill i \in \mathscr{Q} \; , \; j\notin \mathscr{Q} \cr
\displaystyle - \phi_1 \frac{c(u_{i-1} - x_j)}{b(u_{i-1} - x_j)} \frac{a(x_{\bar{\jmath}} - u_{i-1})}{b(x_{\bar{\jmath}} - u_{i-1})} \prod_{\substack{k=1 \\ k \neq i-1}}^{n-1} \frac{a(u_k - u_{i-1})}{b(u_k - u_{i-1})} \lambda_{\mathcal{A}}(u_{i-1}) \cr
\displaystyle \quad\qquad - \; \phi_2 \frac{c(x_j - u_{i-1})}{b(x_j - u_{i-1})} \frac{a(u_{i-1} - x_{\bar{\jmath}})}{b(u_{i-1} - x_{\bar{\jmath}})} \prod_{\substack{k=1 \\ k \neq i-1}}^{n-1} \frac{a(u_{i-1} - u_k)}{b(u_{i-1} - u_k)} \lambda_{\mathcal{D}}(u_{i-1}) \hfill i \notin \mathscr{Q} \; , \; j \in \mathscr{Q} \cr
\displaystyle \phi_1 \frac{a(x_0 - u_{i-1})}{b(x_0 - u_{i-1})} \frac{a(x_1 - u_{i-1})}{b(x_1 - u_{i-1})} \prod_{\substack{k=1 \\ k \neq i-1}}^{n-1} \frac{a(u_k - u_{i-1})}{b(u_k - u_{i-1})} \lambda_{\mathcal{A}}(u_{i-1}) \cr
\displaystyle  \quad\qquad + \; \phi_2  \frac{a(u_{i-1} - x_0)}{b(u_{i-1} - x_0)} \frac{a(u_{i-1} - x_1)}{b(u_{i-1} - x_1)} \prod_{\substack{k=1 \\ k \neq i-1}}^{n-1} \frac{a(u_{i-1} - u_k)}{b(u_{i-1} - u_k)} \lambda_{\mathcal{D}}(u_{i-1}) \hfill i,j \notin \mathscr{Q} \; , \; i = j \cr
\displaystyle - \phi_1 \frac{c(u_{i-1} - u_{j-1})}{b(u_{i-1} - u_{j-1})} \frac{a(x_0 - u_{i-1})}{b(x_0 - u_{i-1})} \frac{a(x_1 - u_{i-1})}{b(x_1 - u_{i-1})} \prod_{\substack{k=1 \\ k \neq i-1, j-1}}^{n-1} \frac{a(u_k - u_{i-1})}{b(u_k - u_{i-1})} \lambda_{\mathcal{A}}(u_{i-1}) \cr
\displaystyle \quad\qquad - \; \phi_2 \frac{c(u_{j-1} - u_{i-1})}{b(u_{j-1} - u_{i-1})} \frac{a(u_{i-1} - x_0)}{b(u_{i-1} - x_0)} \frac{a(u_{i-1} - x_1)}{b(u_{i-1} - x_1)} \prod_{\substack{k=1 \\ k \neq i-1, j-1}}^{n-1} \frac{a(u_{i-1} - u_k)}{b(u_{i-1} - u_k)} \lambda_{\mathcal{D}}(u_{i-1}) \cr
\hfill i,j \notin \mathscr{Q} \; , \; i \neq j \cr
\end{cases} \nonumber \\
\>

\section{Functions $\phi_{\mathrm{H}}$ and $\phi_{\pm}$ } \label{app:PHI}

For completeness reasons we present here explicit expressions for the functions $\phi_{\mathrm{H}}$ and $\phi_{\pm}$ discussed in
\Appref{sec:gen2}. They read as follows:
\begin{dmath*}
\phi_{\mathrm{H}} (x, \Lambda) \coloneqq 2 \left[\left(-2 \left(u_1-x\right) \lambda_{+}(x){}^2-2 \left(u_1-x\right) \left(x \left(2 u_1-x\right) \lambda_{-}'(x)-\Lambda \right) \lambda_{+}(x)+x \Lambda  \left(2 u_1-x\right) \left(2 \left(u_1-x\right) \lambda_{-}'(x)-\lambda_{+}'(x)\right) -\lambda_{+}(u_1 ) \left(2 \left(x^2-2 u_1 x+2 u_1^2\right) \lambda_{-}(x){}^2 \\
+2 \lambda_{-}(x) \left(-\Lambda \left(x^2-2 u_1 x+2 u_1^2\right)
+\lambda_{+}(x) \left(x^2-2 u_1 x+2 u_1^2\right)+x \left(x^2-3 u_1 x+2 u_1^2\right) \lambda_{+}'(x)\right)\right) \lambda_{-}(u_1 ){}^3 +\left(2 \Lambda  \left(-2 x^3+7 u_1 x^2-9 u_1^2 x-x+4 u_1^3+2 u_1\right)-2 \left(4 u_1^3-9 x u_1^2+\left(7 x^2+3\right) u_1-2 \left(x^3+x\right)\right) \lambda_{+}(x) \\
+x \left(3 x^3-14 u_1 x^2+19 u_1^2 x+2 x-8 u_1^3-2 u_1\right) \lambda_{+}'(x)\right) \lambda_{-}(x)+\left(4 x^2-10 u_1 x+5 u_1^2+1\right) \lambda_{+}(x){}^2
-\lambda_{+}(x) \left(\Lambda  \left(4 x^2-10 u_1 x+5 u_1^2+1\right)+x \left(-8 u_1^3+19 x u_1^2-2 \left(7 x^2+1\right) u_1+x \left(3 x^2+2\right)\right)  \lambda_{-}'(x)\right) +x \Lambda \left(\left(-8 u_1^3+19 x u_1^2-2 \left(7 x^2+2\right) u_1+3 \left(x^3+x\right)\right) \lambda_{-}'(x)+\left(2 x^2-7 u_1 x+5 u_1^2+1\right) \lambda_{+}'(x)\right)\right) \lambda_{-}(u_1 ){}^2 \\
+\lambda_{+}(u_1 ){}^2 \left(\left(-2 \left(2 u_1^3-3 x u_1^2+\left(x^2+2\right) u_1-x\right) \lambda_{-}(x){}^2 -2 \left(2 u_1^3-6 x u_1^2+5 x^2 u_1+u_1-x \left(x^2+1\right)\right) \right) \\
+\left(x \left(x^2-6 u_1 x+5 u_1^2+3\right) \lambda_{+}'(x) \left(u_1-x\right){}^2 -\Lambda \left(x^4+4 x^2-16 u_1^3 x-4 \left(2 x^2+3\right) u_1 x+5 u_1^4 \\
+6 \left(3 x^2+1\right) u_1^2+1\right)+\left(x^4+5 x^2-16 u_1^3 x-8 \left(x^2+2\right) u_1 x+5 u_1^4+9 \left(2 x^2+1\right) u_1^2+2\right) \lambda_{+}(x)\right) \lambda_{-}(x)
-\left(u_1-x\right) \lambda_{+}(x) \left(-2 \Lambda \left(x^2-4 u_1 x+2 u_1^2+1\right)-x \left(x^3-7 u_1 x^2+11 u_1^2 x+3 x-5 u_1^3-3 u_1\right) \lambda_{-}'(x)\right)+x \Lambda  \left(\left(x^4+4 x^2-16 u_1^3 x-2 \left(4 x^2+5\right) u_1 x+5 u_1^4+6 \left(3 x^2+1\right) u_1^2+1\right) \lambda_{-}'(x) \\
+\left(x^3-6 u_1 x^2+9 u_1^2 x+2 x-4 u_1^3-2 u_1\right) \lambda_{+}'(x)\right)\right) \lambda_{-}(u_1 ) - \lambda_{+}(u_1 ){}^3 \left(\left(u_1^4-2 x u_1^3+\left(x^2+2\right) u_1^2-2 x u_1-x^2+1\right) \lambda_{-}(x){}^2 \\
+\lambda_{-}(x) \left(-x \left(u_1^2-x u_1+1\right) \lambda_{+}'(x) \left(u_1-x\right){}^3+\Lambda \left(u_1^5-4 x u_1^4+\left(6 x^2+2\right) u_1^3-2 x \left(2 x^2+3\right) u_1^2 \\
+\left(x^4+4 x^2+1\right) u_1-2 x\right)-\left(u_1^5-4 x u_1^4+\left(6 x^2+3\right) u_1^3-4 x \left(x^2+2\right) u_1^2+\left(x^4+5 x^2+2\right) u_1-2 x\right) \lambda_{+}(x)\right) -\left(u_1-x\right) \left(-\left(u_1^3-3 x u_1^2+2 x^2 u_1+u_1-x\right) \lambda_{+}(x){}^2-\left(u_1-x\right) \left(x \left(u_1^3-2 x u_1^2+x^2 u_1+u_1-x\right) \lambda_{-}'(x) \\
- \Lambda \left(u_1^2-2 x u_1+1\right)\right) \lambda_{+}(x)+x \Lambda  \left(u_1^2-x u_1+1\right) \left(\left(x^2-2 u_1 x+u_1^2+1\right) \lambda_{-}'(x)+\left(x-u_1\right) \lambda_{+}'(x) \right)\right)\right)\right]  \\
\left[ \left(\lambda_{-}(u_1 )+u_1 \lambda_{+}(u_1 ) \right) \left(2 \left(u_1-x\right) \lambda_{-}(u_1 )+\left(x^2-2 u_1 x+u_1^2+1\right) \lambda_{+}(u_1 )\right)   \left(\lambda_{-}(u_1) \left(2 \left(u_1-x\right) \lambda_{-}(x)-\lambda_{+}(x)\right)+\left(\left(x^2-2 u_1 x+u_1^2+1\right) \lambda_{-}(x)+\left(x-u_1\right) \lambda_{+}(x)\right) \lambda_{+}(u_1 ) \right)\right]^{-1}
\end{dmath*}

\begin{dmath*}
\phi_{+} (x, \Lambda) \coloneqq \left[\left(8 \lambda_{-}(x){}^2 \left(u_1-x\right){}^3-2 \left(2 x^3+4 u_1^2 x+x-\left(6 x^2+1\right) u_1\right) \lambda_{+}(x){}^2
+2 x \left(x^2-3 u_1 x+2 u_1^2\right) \lambda_{+}(x) \left(2 \Lambda +x \left(x-2 u_1\right) \lambda_{-}'(x)\right) 
-x^2 \Lambda \left(x-2 u_1\right){}^2 \left(\lambda_{+}'(x)-2 \left(u_1-x\right) \lambda_{-}'(x)\right)
+2 \lambda_{-}(x) \left(\left(8 x u_1^3-4 \left(4 x^2+1\right) u_1^2+4 x \left(3 x^2+2\right) u_1-x^2 \left(3 x^2+4\right)\right) \lambda_{+}(x) \\
+x \left(2 u_1-x\right) \left(\Lambda  \left(-3 x^2+6 u_1 x-4 u_1^2\right)+x \left(x^2-3 u_1 x+2 u_1^2\right) \lambda_{+}'(x)\right)\right)\right) \lambda_{-}(u_1 ){}^3+\lambda_{+}(u_1 ) \left(\Lambda \left(2 u_1-x\right) \left(\left(6 u_1^3-13 x u_1^2+\left(8 x^2+6\right) u_1-x \left(x^2+5\right)\right) \lambda_{-}'(x) \\
-\left(x^2-5 u_1 x+4 u_1^2+2\right) \lambda_{+}'(x)\right) x^2 
-\lambda_{+}(x) \left(\Lambda  \left(-16 u_1^3+34 x u_1^2-4 \left(5 x^2+2\right) u_1+3 x \left(x^2+2\right)\right)+x \left(12 u_1^4-32 x u_1^3+\left(29 x^2+8\right) u_1^2 \\
-2 x \left(5 x^2+6\right) u_1+x^2 \left(x^2+4\right)\right) \lambda_{-}'(x)\right) x+2 \left(6 u_1^4-24 x u_1^3+\left(28 x^2+6\right) u_1^2-12 \left(x^3+x\right) u_1+x^2\left(x^2+6\right)\right) \lambda_{-}(x){}^2+\left(3 x^4+7 x^2-16 u_1^3 x-2 \left(10 x^2+7\right) u_1 x+\left(34 x^2+5\right) u_1^2+1\right) \lambda_{+}(x){}^2+\lambda_{-}(x) \left(2 \left(12 x u_1^4-8 \left(4 x^2+1\right) u_1^3+\left(30 x^3+32 x\right) u_1^2-\left(11 x^4+30 x^2+4\right) u_1 \\
+x \left(x^4+8 x^2+4\right)\right) \lambda_{+}(x)+x \left(x \left(12 u_1^4-32 x u_1^3+\left(29 x^2+8\right) u_1^2-2 x \left(5 x^2+6\right) u_1+x^2 \left(x^2+4\right)\right) \lambda_{+}'(x) \\
-2 \Lambda \left(12 u_1^4-32 x u_1^3+6 \left(5 x^2+2\right) u_1^2-x \left(11 x^2+16\right) u_1+x^2 \left(x^2+6\right)\right)\right) \right)\right) \lambda_{-}(u_1 ){}^2
+\lambda_{+}(u_1 ){}^2 \left(\Lambda  \left(-u_1^2+x u_1-1\right) \left(\left(2 x^2-7 u_1 x+5 u_1^2+1\right) \lambda_{+}'(x)-2 \left(3 u_1^3-7 x u_1^2+\left(5 x^2+3\right) u_1 \\
-x \left(x^2+2\right)\right) \lambda_{-}'(x)\right) x^2-2 \lambda_{+}(x) \left(x \left(3 u_1^5-10 x u_1^4+4 \left(3 x^2+1\right) u_1^3-3 x \left(2 x^2+3\right) u_1^2+\left(x^4+6 x^2+1\right) u_1-x \left(x^2+1\right)\right) \lambda_{-}'(x) \\
-\Lambda  \left(5 u_1^4-14 x u_1^3+6 \left(2 x^2+1\right) u_1^2-x \left(3 x^2+8\right) u_1+2 x^2+1\right)\right) x+2 \left(3 u_1^5-15 x u_1^4+\left(22 x^2+6\right) u_1^3-6 x \left(2 x^2+3\right) u_1^2+\left(2 x^4+10 x^2+3\right) u_1-3 x\right) \lambda_{-}(x){}^2-2 \left(5 x u_1^4-2 \left(7 x^2+1\right) u_1^3+3 x \left(4 x^2+3\right) u_1^2-\left(3 x^4+9 x^2+1\right) u_1+2 x^3+x\right) \lambda_{+}(x){}^2 \\
+2 \lambda_{-}(x) \left(\left(6 x u_1^5-5 \left(4 x^2+1\right) u_1^4+4 x \left(6 x^2+7\right) u_1^3-3\left(4 x^4+13 x^2+2\right) u_1^2 \\
+2 x \left(x^4+8 x^2+7\right) u_1-x^4-5 x^2-1\right) \lambda_{+}(x)+x \left(\Lambda  \left(-6 u_1^5+20 x u_1^4-12 \left(2 x^2+1\right) u_1^3+12 x \left(x^2+2\right) u_1^2 \\
-2 \left(x^4+6 x^2+3\right) u_1+x \left(x^2+4\right)\right)+x \left(3 u_1^5-10 x u_1^4+4 \left(3 x^2+1\right) u_1^3-3 x \left(2 x^2+3\right) u_1^2 \\
+\left(x^4+6 x^2+1\right) u_1-x \left(x^2+1\right)\right) \lambda_{+}'(x)\right)\right)\right) \lambda_{-}(u_1 )
+\lambda_{+}(u_1 ){}^3 \left(x^2 \Lambda  \left(\left(x^2-2 u_1 x+u_1^2+1\right) \lambda_{-}'(x)+\left(x-u_1\right) \lambda_{+}'(x)\right) \left(u_1^2-x u_1+1\right){}^2 \\
-x \lambda_{+}(x) \left(x \left(u_1^2-x u_1+1\right) \lambda_{-}'(x) \left(u_1-x\right){}^2+\Lambda \left(-2 u_1^3+5 x u_1^2-\left(3 x^2+2\right) u_1+x\right)\right) \left(u_1^2-x u_1+1\right) \\
+\left(u_1^6-6 x u_1^5+\left(11 x^2+3\right) u_1^4-4 x \left(2 x^2+3\right) u_1^3+\left(2 x^4+10 x^2+3\right) u_1^2-6 x u_1-x^2+1\right) \lambda_{-}(x){}^2 \\
-\left(u_1-x\right){}^2 \left(2 x u_1^3-\left(3 x^2+1\right) u_1^2+4 x u_1-1\right) \lambda_{+}(x){}^2 + \lambda_{-}(x) \left(2 \left(x u_1^6-\left(4 x^2+1\right) u_1^5+x \left(6 x^2+7\right) u_1^4 \\ 
- \left(4 x^4+13 x^2+2\right) u_1^3+x \left(x^4+8 x^2+7\right) u_1^2-\left(x^4+5 x^2+1\right) u_1+x\right) \lambda_{+}(x) \\ 
+ x \left(u_1^2-x u_1+1\right) \left(x \left(u_1-x\right){}^2 \left(u_1^2-x u_1+1\right) \lambda_{+}'(x)-2 \Lambda \left(u_1^4-3 x u_1^3 \\
+\left(3 x^2+2\right) u_1^2-x \left(x^2+3\right) u_1+1\right)\right)\right)\right) \right] \left[\left(2 u_1 \lambda_{-}(u_1 )+\left(u_1^2+1 \right) \lambda_{+}(u_1 )\right){}^2 \left(2 \left(u_1-x\right) \lambda_{-}(u_1 ) \\
+\left(x^2-2 u_1 x+u_1^2+1\right) \lambda_{+}(u_1 )\right) \left(\lambda_{-}(u_1 ) \left(2 \left(u_1-x\right) \lambda_{-}(x) \\
-\lambda_{+}(x)\right)+\left(\left(x^2-2 u_1 x+u_1^2+1\right) \lambda_{-}(x)+\left(x-u_1\right) \lambda_{+}(x)\right) \lambda_{+}(u_1 )\right)\right]^{-1}
\end{dmath*}

\begin{dmath*}
\phi_{-} (x, \Lambda) \coloneqq - \left[\left(2 u_1 \lambda_{-}(u_1 )+\left(u_1^2+1 \right) \lambda_{+}(u_1 )\right){}^2 \left(\lambda_{+}(u_1 ) \lambda_{-}(u_1 ){}^2 \left(\Lambda  \left(\left(-10 u_1 x+5 u_1^2+5 x^2+1\right) \lambda_{-}'(x)+3 \left(x-u_1\right) \lambda_{+}'(x)\right) \\
+\left(u_1-x\right) \lambda_{-}(x) \left(-2 \Lambda +5 \left(u_1-x\right) \lambda_{+}'(x)+2 \lambda_{+}(x)\right)-\lambda_{+}(x) \left(\Lambda +5 \left(u_1-x\right){}^2 \lambda_{-}'(x)\right) \\
-2 \lambda_{-}(x){}^2+\lambda_{+}(x){}^2\right)+\lambda_{+}(u_1 ){}^2 \lambda_{-}(u_1 ) \left(\left(u_1-x\right) \left(\Lambda  \left(\left(-8 u_1 x+4 u_1^2+4 x^2+2\right) \lambda_{-}'(x)+3 \left(x-u_1\right) \lambda_{+}'(x)\right) \\
-2 \lambda_{+}(x) \left(\Lambda +2 \left(u_1-x\right){}^2 \lambda_{-}'(x)\right)+2 \lambda_{+}(x){}^2\right)+2 \lambda_{-}(x) \left(\Lambda +2 \left(u_1-x\right){}^3 \lambda_{+}'(x)-\lambda _2(x)\right) \\
-2 \left(u_1-x\right) \lambda_{-}(x){}^2\right)+\lambda_{+}(u_1 ){}^3 \left(\left(u_1-x\right){}^2 \left(\Lambda  \left(\left(-2 u_1 x+u_1^2+x^2+1\right) \lambda_{-}'(x)+\left(x-u_1\right)  \lambda_{+}'(x)\right)-\lambda_{+}(x) \left(\Lambda +\left(u_1-x\right){}^2 \lambda_{-}'(x)\right) \\
+\lambda_{+}(x){}^2\right)-\left(-2 u_1 x+u_1^2+x^2-1\right) \lambda_{-}(x){}^2+\left(u_1-x\right) \lambda_{-}(x) \left(2 \Lambda +\left(u_1-x\right){}^3 \lambda_{+}'(x)-2 \lambda_{+}(x)\right)\right)+\lambda_{-}(u_1 ){}^3 \left(-\left(-2 \left(u_1-x\right) \left(\Lambda -\lambda_{+}(x)\right) \lambda_{-}'(x)+2 \lambda_{-}(x) \left(\Lambda +\left(x-u_1\right) \lambda_{+}'(x)-\lambda_{+}(x)\right)+\Lambda  \lambda_{+}'(x)\right)\right)\right) \right] \\
\left[\left(\lambda_{-}(u_1 )+u_1 \lambda_{+}(u_1 )\right){}^2 \left(\lambda _2(u_1 ) \left(-2 u_1 x+u_1^2+x^2+1\right)+2 \lambda_{-}(u_1 ) \left(u_1-x\right)\right) \left(\lambda _2(u_1 ) \left(\left(-2 u_1 x+u_1^2+x^2+1\right) \lambda_{-}(x)+\left(x-u_1\right) \lambda_{+}(x)\right)+\lambda_{-}(u_1 ) \left(2 \left(u_1-x\right) \lambda_{-}(x)-\lambda_{+}(x)\right)\right) \right]^{-1}
\end{dmath*}

\end{document}